\documentclass{IEEEtaes}

\usepackage{color,array,amsthm}

\usepackage{graphicx}
\usepackage{pifont}
\usepackage{relsize}
\usepackage{bbm}
\usepackage{newtxmath}
\usepackage[utf8]{inputenc}
\usepackage{amsfonts}
\usepackage{comment}
\usepackage{bbm}
\usepackage{setspace}
\usepackage[USenglish]{babel}
\usepackage{multirow}
\usepackage[super]{nth}
\usepackage[en-GB]{datetime2}
\usepackage{float}
\usepackage{arabtex}
\usepackage{utf8}
\usepackage{changepage}
\usepackage{diagbox}
\usepackage{xcolor}
\usepackage{tikz}
\usepackage{units}
\usepackage{booktabs}
\DTMlangsetup[en-US]{ord=raise}
\theoremstyle{plain}
\newtheorem{corollary}{Corollary}  
\newtheorem{remark}{Remark} 

\allowdisplaybreaks

\usepackage{cite}

\ifCLASSINFOpdf
  \graphicspath{{./Figures/}}
  \DeclareGraphicsExtensions{.pdf,.jpeg,.png}
\else
\fi

\usepackage{amsmath}

\usepackage{algorithmic}
\usepackage{dblfloatfix}
\usepackage{url}

\newtheorem{lemma}{Lemma}
\begin{document}
\renewcommand{\figurename}{Fig.}

\title{Orthogonality Analysis in LoRa Uplink Satellite Communications Affected by Doppler Effect}

\author{JIKANG DENG}
\member{Student Member, IEEE}
\affil{King Abdullah University of Science and Technology, Jeddah, SA} 

\author{FATMA BENKHELIFA}
\member{Member, IEEE}
\affil{Queen Mary University of London, London, UK.} 

\author{MOHAMED-SLIM ALOUINI}
\member{Fellow, IEEE}
\affil{King Abdullah University of Science and Technology, Jeddah, SA}

\authoraddress{Jikang Deng and Mohamed-Slim Alouini are with CEMSE Division, 
King Abdullah University of Science and Technology (KAUST),
Thuwal, 23955-6900, Saudi Arabia,
(e-mail: \href{mailto:jikang.deng@kaust.edu.sa}{jikang.deng@kaust.edu.sa};
\href{mailto:slim.alouini@kaust.edu.sa}{slim.alouini@kaust.edu.sa}); Fatma Benkhelifa is with the School of Electronic Engineering and Computer Science (EECS),
Queen Mary University of London (QMUL), 
Mile End Road, London E1 4NS, United Kingdom, (e-mail: \href{mailto:f.benkhelifa@qmul.ac.uk}{f.benkhelifa@qmul.ac.uk}). 
\itshape (Corresponding author: Fatma Benkhelifa)}

\maketitle

\begin{abstract}
This paper provides, for the first time, analytical expressions for the Long-Range (LoRa) waveform and cross-correlation in both continuous and discrete time domains under the Doppler effect in satellite communication. We propose the concept and formulas of the shared visibility window for satellites toward two ground devices. Our analysis covers cross-correlation results with varying spreading factors (SF) and bandwidths for both no-Doppler and with-Doppler cases.
We find that the maximum cross-correlation with different SFs and the mean cross-correlation are immune to the Doppler effect. However, the maximum cross-correlation with the same SFs is only immune to high Doppler shift, with its value fluctuating between 0.6 and 1 under the condition of high Doppler rate. We interpret this fluctuation by introducing the relationship between transmission start time and cross-correlation.
We provide a parameter analysis for orbit height, ground device distance, and inclination angle. Finally, we analyze the bit error rate (BER) for LoRa signals and observe worse performance under the high Doppler shift or interference with the same SF. We observe that increasing the signal-to-noise ratio (SNR) or the signal-to-interference ratio (SIR) improves the BER only when the Doppler effect is below a frequency threshold. Notably, under the Doppler effect, the performance behaviors of BER no longer align with those of maximum cross-correlation.
Finally, our results lead to two recommendations: 1) To mitigate the Doppler impact on cross-correlation, we recommend utilizing low SFs, high orbit height, short ground device distance, and the transmission start time with high Doppler shift; 2) To mitigate the Doppler impact on BER, we recommend employing low SFs, high bandwidth, and transmission start time with high Doppler rate. These conflicting recommendations regarding transmission start time highlight the necessity of Doppler shift compensation techniques to help operate LoRa in space properly.
}
\end{abstract}



\section{INTRODUCTION}
T{\scshape he} Internet of Things (IoT) is an emerging concept that will continuously have a significant impact on the socio-economic growth of nations. It encompasses a wide range of applications that touch different areas of our daily lives, such as smart agriculture, healthcare, and asset management. IoT is a large network consisting of various devices or sensors that interact and communicate with each other. Through collecting and analyzing data from the devices, the IoT system will draw conclusions and generate more scientific decisions for making everyone's life more efficient and convenient.

Low power wide area networks (LPWANs), which emerged in 2010, are mainly used for achieving massive IoT and critical IoT connectivity due to their wide coverage range and great energy efficiency \cite{chettri2019comprehensive}. Up to now, there are several promising LPWANs technologies, such as Long-Range (LoRa), SigFox, narrow-band IoT (NB-IoT), and so on. LoRa, a technology developed by Cycleo and acquired by Semtech now, is competitive for modeling IoT systems that operate over long-range distances with very low power consumption \cite{Semtech_LoRa}. To achieve long range, a correlation mechanism is used based on band spreading methods. This mechanism ensures even extremely small signals that may be lost in the noise can be accurately demodulated by the receiver. Meanwhile, LoRa receivers are still able to decode signals, which are up to $19.5$ dB below the noise. With those distinguished advantages, LoRa devices are becoming promising supports for emerging applications such as water and gas metering, cold chain monitoring, and traffic optimization. It is estimated that LoRa will be the leading non-cellular LPWANs technology by 2026, and will account for over $50$\% of all non-cellular LPWANs \cite{krishnan2020lorawan}.

As a physical layer (PHY) implementation, LoRa is based on a spread spectrum modulation technique derived from chirp spread spectrum (CSS) technology and mixed with frequency shift keying (FSK) \cite{vangelista2017frequency,bor2016lora}, and that is patented by Semtech. During the working process, the initial serial data will experience error correction encoding based on Hamming codes and coding rate, whitening, interleaving, and CSS modulation before being sent. What's more, LoRa uses two options of packet header format, explicit and implicit, for data transmission, described in \cite{LoRaWAN}. Its range of communication can reach up to $5$ km in urban areas and $15$ km or much more in rural areas. It is also able to be operated on $2.4$ GHz to achieve higher data rates compared to sub-gigahertz bands, at the expense of range.

LoRa wide area network (LoRaWAN) is a Media Access Control (MAC) layer protocol developed by LoRa Alliance and opened to the public. It is officially approved as a standard for LPWANs by the International Telecommunication Union (ITU) \cite{LoRaWAN}. The channel access mechanism is based on an ALOHA-like channel access technique, described in \cite{el2018lora}.  Considering the different needs of IoT applications for energy autonomy and data exchange, it also defines three options for scheduling the receive window slots for downlink communication: class A, B, and C \cite{petajajarvi2017performance}. 

As a significant part of 6G communication, satellite communication needs to establish a mature and efficient network. LoRa has been suggested as one of the best candidates for satellite communication because it meets many requirements of satellite IoT scenarios, and has low power consumption and long-range propagation \cite{kiki2022improved,wu2019research}. What's more, LoRa over satellite can help achieve some of the goals of 6G such as global connectivity.
It has already been used in the launches of some satellite companies such as Lacuna, Actility, Wyld Networks, etc. 
Furthermore, Low-Earth-Orbit (LEO) satellites, especially CubeSat nanosatellites, will be the main carriers for LoRa technology due to their low altitude and short propagation time. However, we should notice that low altitude will lead to the high velocity of the satellite, causing high Doppler shifts and difficulties in demodulating the received signals. It is claimed that LoRa is known to be immune to Doppler shift \cite{liando2019known}. However, there are still not enough detailed criteria for the applicability of the LoRa modulation in terms of orthogonality under strong Doppler shift in satellite communication, not to mention the case of high Doppler rate.
Therefore, it is meaningful to carry out a precise analysis of the immunity of LoRa modulation in the presence of Doppler shift, then we can explore the potential method to extend LoRa's ability to cope with the Doppler shift, and improve the performance of satellite communication based on LoRa technology.

\subsection{Related Work}
LPWANs technologies have captured tremendous research interest in investigating the performance of different technologies under interference from the LEO satellite scenario. \cite{temim2022low,wang2022performance} concluded the feasibility of deployment with LEO satellite and performance of LoRa, SigFox, and NB-IoT. All of them can reach a $1000$ km communication range under specific settings, satisfying LEO satellite's needs. Specifically, LoRa is based on a license-free band with bandwidth such as $125$ kHz or $250$ kHz and LoRa data rate can range from as low as 0.3 kbit/s to around 50 kbit/s per channel, depending on the configuration of the spreading factor and bandwidth \cite{LoRaAlliance2020}.

As for the immunity of LoRa toward the Doppler effect, some recent research studies analyzed its sensitivity in mobility scenarios, such as vehicle-to-everything (V2X) scenario \cite{torres2021experimental, li2018lora}
and satellite-to-Earth communication scenarios \cite{doroshkin2018laboratory,doroshkin2019experimental,zadorozhny2022first,cao2021influence,colombo2022low,ameloot2021characterizing,ullah2023understanding,lapapan2021lora}.  
\cite{torres2021experimental, li2018lora} showed that LoRa is sensitive to the Doppler effect due to the velocity and concluded that lower $SF$ has less sensitivity to the effect due to a higher reception threshold based on the results of packet delivery rate.
\cite{doroshkin2018laboratory,doroshkin2019experimental} proved that LoRa holds immunity to the Doppler effect for orbit with altitudes higher than $550$ km. \cite{zadorozhny2022first} proposed that under $SF \leq 11$ and bandwidth $B> 31.25$ kHz, LoRa has high immunity towards the Doppler effect, based on the results obtained from NORBY Cubesat operating at $560$ km.
\cite{cao2021influence} concluded that when $SF$ is greater than $9$, within the range of the maximum Doppler-Rate, the packet error rate will reach $100$\% as the Doppler rate increases.  
Among them, some works were carried out with TLE dataset from the satellites \cite{zadorozhny2022first}, while other researchers set up some similar scenarios or computer environments to simulate the satellite-to-Earth communication \cite{doroshkin2018laboratory,doroshkin2019experimental,cao2021influence}.

Other research papers discussed and distinguished the impact of high Doppler shift or Doppler rate. \cite{colombo2022low} proposed a low-cost SDR method for evaluating LoRa satellite communication. 
\cite{ameloot2021characterizing} evaluates the impact of Doppler shift on LoRa in body-centric applications and analyzes the performance of symbol detection levels such as bit error rate (BER) or symbol error rate (SER). 
\cite{ullah2023understanding} analyzed the impact of the high Doppler shift or Doppler rate based on packet losses of LoRa in a Direct-to-Satellite (DtS) connectivity scenario and offered suggestions for the selection of suitable parameters to mitigate the Doppler effect based on a parameter analysis. They further demonstrated that changes in $SF$, payload length, and LoRaWAN low data rate optimization (LDRO) impact the vulnerability to the high Doppler rate more than that to the high Doppler shift.

\begin{table*}[ht]
    \centering
    \renewcommand{\arraystretch}{1.2} 
    \caption{Related Work on Evaluation of LoRa Modulation or LoRa Satellite Communication}
    \label{tab:related_work}
    \scriptsize
    \begin{tabular}{|c|p{1.5cm}|c|c|p{1cm}|p{1cm}|p{1.3cm}|p{7cm}|}
        \hline
        \multirow{2}{*}{\centering \textbf{Ref.}} & 
        \multirow{2}{*}{\parbox{2.5cm}{\textbf{Evaluation}\\\textbf{Index}}} &
        \multirow{2}{*}{\parbox{0.7cm}{\textbf{Doppler}\\\textbf{Effect}}} & 
        \multirow{2}{*}{\parbox{0.7cm}{\textbf{Inter-}\\\textbf{ference}}} & \multirow{2}{*}{\parbox{1cm}{\textbf{Model}\\\textbf{Dimension}}}  & \multicolumn{2}{c|}{\textbf{Parameter Summary}} & \multirow{2}{*}{\textbf{{Descirption}}}\\
        \cline{6-7}
        & & & & & \textbf{SF} &  \textbf{Height (km)} & \\
        \hline
        \cite{doroshkin2018laboratory} & Radio link immunity margin & \ding{51} & \ding{55} & 2D & 7, 11 & 200 &  A laboratory testing to show the immunity of LoRa radio link to Doppler shift\\ 
        \hline
        \cite{doroshkin2019experimental} & Packet success ratio & \ding{51} & \ding{55} & 2D & 7, 11, 12 & 200\,-\,550  & The laboratory testing and outdoor experiments on the feasibility of LoRa modulation in satellite communication under Doppler effect\\ 
        \hline
        \cite{zadorozhny2022first} & Packet success number & \ding{51} & \ding{55} & - & 7, 10, 11, 12 &  560 & The flight tests of LoRa modulation for robustness against the Doppler effect in the NORBY CubeSat based satellite-to-Earth radio channel \\ 
        \hline
        \cite{cao2021influence} & PER & \ding{51} & \ding{55} & 2D & 7\,-\,12 & 600 & Analysis of the influence of the Doppler effect on LoRa in satellite-to-Earth radio channel, and a Doppler Rate estimation algorithm  \\ 
        \hline
        \cite{colombo2022low} & PER & \ding{51} & \ding{55} & 2D & 7\,-\,12 & 200, 550 & A low-cost SDR-based tool for emulating Doppler effect and evaluate LoRa satellite communication \\ 
        \hline
        \cite{ameloot2021characterizing} & BER, SER & \ding{51} & \ding{55} & 2D & 7\,-\,12 & - & Assessment of Doppler effect on LoRa symbol detection performance based on experiments with an SDR-based LoRa symbol detector.\\ 
        \hline
        \cite{ullah2023understanding} & PDR & \ding{51} & \ding{55} & 2D & 7, 10, 12 & 560,\,750,\,1000 & Analysis of LoRa DtS performance under Doppler effect in LEO scenario \\ 
        \hline
        \cite{lapapan2021lora} & RSSI & \ding{51} & \ding{55} & 2D & - & 500,\,550,\,600 & Theoretical investigation and laboratory test of LoRa multi-channel access immunity to Doppler effect in satellite communication \\ 
        \hline
        \cite{croce2017impact} & PER, Data extraction rate  & \ding{55} & \ding{51} & - & 6\,-\,12 & -  &  A study of the impact of SF imperfect orthogonality on LoRa communication \\ 
        \hline
        \cite{benkhelifa2022orthogonal} & Orthogonality,
        BER, PDR & \ding{55} & \ding{51} & - & 5\,-\,12 & - & A study of the orthogonality of LoRa modulation  \\ 
        \hline
        Ours & Orthogonality, BER & \ding{51} & \ding{51} & 3D & 5\,-\,12  & 200\,-\,900 & A study of the orthogonality of LoRa uplink satellite communication under Doppler effect  \\ 
        \hline
    \end{tabular}
\end{table*}

However, the aforementioned related works demonstrated various research gaps. 
\cite{doroshkin2018laboratory, doroshkin2019experimental, zadorozhny2022first, ullah2023understanding} lacked a thorough analysis of LoRa in terms of different $SF$. \cite{doroshkin2018laboratory, zadorozhny2022first, cao2021influence,colombo2022low, ullah2023understanding, lapapan2021lora} only focused on a specific setting of orbit height, which cannot be generalized. 
The indexes used for evaluating the performance of LoRa in the presence of Doppler effect are limited to radio link immunity margin \cite{doroshkin2018laboratory}, packet delivery ratio (PDR) or packet success ratio, \cite{doroshkin2019experimental,zadorozhny2022first,ullah2023understanding}, SER or BER \cite{ameloot2021characterizing}, packet error rate (PER) \cite{cao2021influence,colombo2022low}, and received signal strength indication (RSSI) \cite{lapapan2021lora}. There is still no detailed analysis focusing on the LoRa interference and the orthogonality or cross-correlation of LoRa waveforms. \cite{benkhelifa2022orthogonal,croce2017impact} had analyzed the orthogonality or cross-correlation of LoRa and proved LoRa to be nonorthogonal, but they were based on the scenario without the Doppler effect. What's more, to generalize the conclusions, the system model containing the ground device and the satellite should be determined on a three-dimensional (3D) plane. But \cite{doroshkin2019experimental, doroshkin2018laboratory, cao2021influence, colombo2022low, ameloot2021characterizing, ullah2023understanding, lapapan2021lora} modeled the Doppler effect on a two-dimensional (2D) plane, assuming that the satellite can always pass directly above the ground device. In addition, 
\cite{colombo2022low,ameloot2021characterizing,ullah2023understanding} provided questionable conclusions on the Doppler shift effect on LoRa performance. 
\cite{colombo2022low} showed that LoRa performance on PER is immune to the high Doppler shift, but not to the high Doppler rate. \cite{ameloot2021characterizing} confirmed that BER or SER are sensitive to the high Doppler shift, which can be eliminated if the Doppler shift does not exceed $10$\% of the bandwidth. 
\cite{ullah2023understanding} showed that the impact of high Doppler shift or Doppler rate on PDR depends on various parameters ($SF$, payload length, orbit height, etc.). To clearly summarize the contributions and limitations of above related work, we provide Table \ref{tab:related_work}, with a particular focus on the Doppler effect, interference, model dimension, and some related parameters.


Therefore, considering the limitations of related works, this paper will systematically analyze the immunity of LoRa modulation under the Doppler effect in a general 3D satellite-to-Earth communication scenario given various parameters ($SF$, orbit height, bandwidth, ground device distance, etc.). What's more, this paper will focus on analyzing the impact of another LoRa interference signal, both high Doppler shift and high Doppler rate respectively on the orthogonality criteria and BER performance. 

\subsection{Contributions}
The main contributions of our work are: 

\begin{itemize}
\item We provide the analytical expressions for the satellite's motion and the Doppler shift in a 3D communication scenario with two stationary ground devices and one LEO satellite. It is worth noting that this model is generalizable to any possible orbit height, inclination, or position of ground devices. We define the shared visibility window between a satellite and two ground devices and provide the approximation method for Doppler shift with high accuracy.
\item We provide, for the first time, the general expressions of LoRa waveform and cross-correlation functions under Doppler shift in both continuous and discrete time domains. We conclude the cases when the immunity towards Doppler shift is acquired. To explain how a certain cross-correlation varies under the Doppler shift, we have designed a theoretical explanation strategy, which introduces the relationship between transmission start time, differential Doppler shift, carrier frequency difference, and cross-correlation.
\item We provide the parameter analysis based on the sensitivity of cross-correlation towards the Doppler effect under different orbit heights, inclinations, and ground device distances. We also characterize the impact of the interference signal through the BER results. We finally offer parameter configuration suggestions separately for improving the immunity of cross-correlation or BER towards the Doppler effect based on the parameter analysis.
\item We compare cross-correlation and BER results under high Doppler shift or high Doppler rate. Furthermore, we suggest the necessity of Doppler shift compensation techniques to ensure the good performance of LoRa in satellite communication.
\end{itemize}

The remainder of the paper is organized as follows. Section \ref{System model} introduces the system model containing LoRa modulation and the satellite's motion. Section \ref{Doppler_characteriztion} characterized the Doppler shift, shared visibility window, and the Doppler approximation method. Section \ref{Cross-correlation_functions} shows the details of analytical cross-correlation expressions under the Doppler effect in both continuous and discrete time domains. Section \ref{Numerical_results} provides the cross-correlation results in no-Doppler and with-Doppler cases, the parameter analysis, the BER results, and the corresponding suggestions for parameter choosing. The conclusion of this paper is given in Section \ref{conclusion}.

\section{SYSTEM MODEL}
\label{System model}
In this paper, we consider the uplink transmissions from two stationary LoRa ground devices, communicating with one moving LEO satellite equipped with LoRa transceiver. Fig. \ref{orbit} illustrated the satellite's motion geometry around the LoRa ground devices. Our target is to characterize the cross-correlation functions of signals transmitted by these two LoRa ground devices to the satellite.
Generally, for each satellite orbit, there are six orbital elements: semi‑major axis of the satellite orbit, eccentricity, mean anomaly, inclination, right ascension of the ascending node, and argument of periapsis\cite{rouzegar2019estimation}. Those elements will determine the orbit's shape, size, orientation, and the satellite's location. In our scenario, we assume the satellite is in the type of Walker Constellation, where the orbit is circular. \cite{doroshkin2019experimental,lapapan2021lora}.
\begin{figure}[ht]
  \centering
  \includegraphics[width = 8cm]{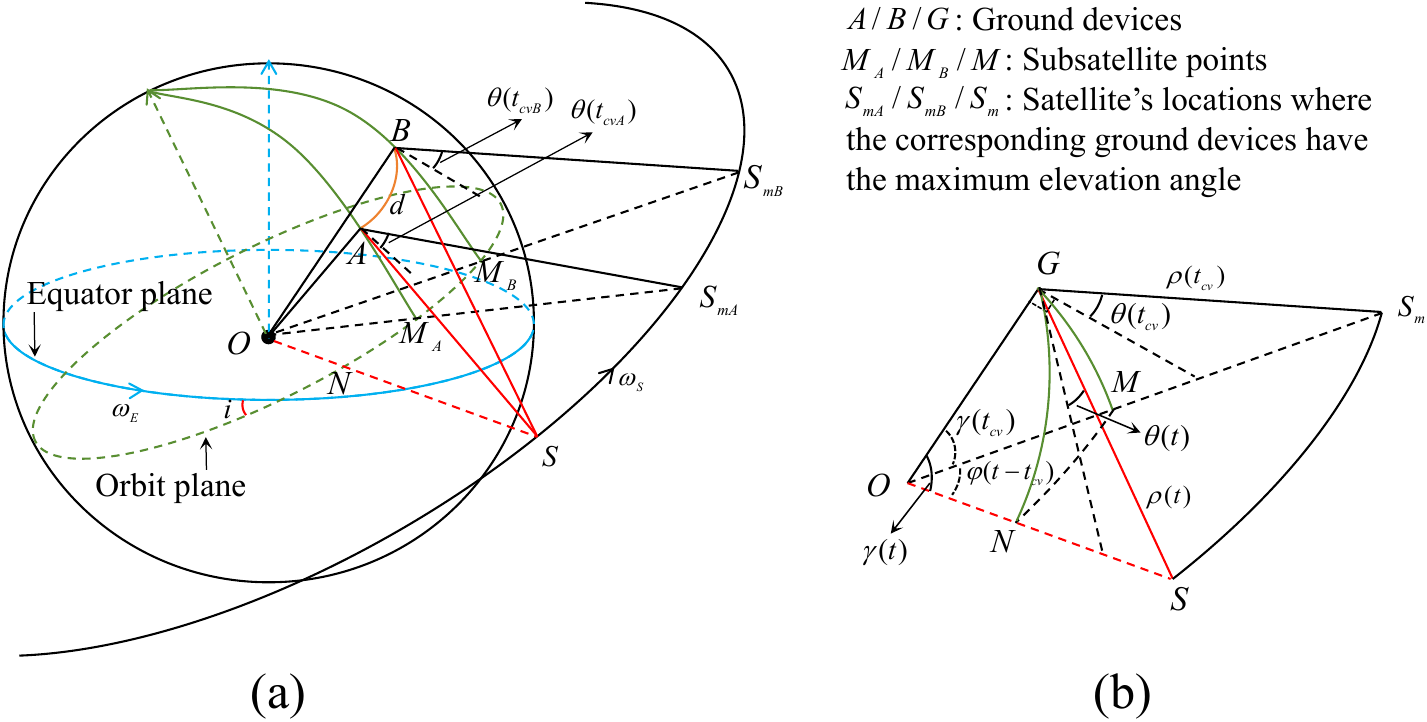}
  \caption{Satellite's motion geometry towards ground devices: (a) Overall geometry structure; (b) Detailed geometry analysis}
  \label{orbit}
\end{figure}
\subsection{LoRa Waveform Under Doppler Shift}
The satellite and the ground devices use LoRa technology. LoRa is known for its long communication range and low power consumption. It is developed to operate in the licence-free frequency spectrum which varies from one region to another region. For example, China uses the CN470-510 MHz, Europe uses the EU863-870 MHz, and United States uses the US902-928 MHz, etc. Meanwhile, up to now, LoRa is designed to configure eight $SF$s, ranging from $5$ to $12$, which represents the number of bits to be transmitted in one symbol. Thus, the symbol value of one waveform will range from $0$ to $2^{SF}\!-\!1$. The common bandwidths for LoRa are $125$ kHz, $250$ kHz, and $500$ kHz. There are four coding rates used in LoRa: $4/5$, $4/6$, $5/7$, and $4/8$. The coding rate is used for adding the redundant bits to the data to be transmitted, helping to realize the forward error correction (FEC) \cite{Semtech_LoRa,LoRaWAN,bor2016lora}. 
Based on CSS modulation, the frequency of LoRa is time-varying and bounded by its bandwidth. The frequency and waveform functions without Doppler shift have been given by \cite{benkhelifa2022orthogonal,chiani2019lora}.

In this scenario, we define the Doppler shift as $f_d(t)$ or $f_d(n)$ for continuous and discrete time domains respectively, whose analytical derivation and expression will be provided in the following sections. Therefore, we can derive the frequency and waveform functions under Doppler shift in both continuous and discrete time domains.

\textit{Continuous-time Domain:} The frequency function when given symbol $k$ is 
\begin{equation}
  \label{lora-con-f}
  \begin{aligned} 
    f_{k,d}(t)  & \!=\!\begin{cases}
      f_{\min}\!+\!f_d(t)\!+\!\frac{B}{2^{SF}}(\frac{t-t_0}{T} +k),
      \\ \hfill \hfill \text{if }0\le t-t_0 \le t_k, \\
      f_{\min}\!+\!f_d(t)\!+\!\frac{B}{2^{SF}}(\frac{t-t_0}{T}+k\!-\!2^{SF}), \\ \hfill \hfill \text{if }t_k<  t-t_0 \le T_s.
    \end{cases}
  \end{aligned}
\end{equation}
Here $f_{\min}$ is the minimum frequency of LoRa signal and $f_{\min} = f_c - \frac{B}{2}$, where $f_c$ is the carrier frequency, and $B$ is the bandwidth, $T_s$ is the symbol period which can be obtained by $T_s = 2^{SF}T$ and $T = 1/B$. In addition, $t_0$ is the starting time of LoRa signal, $k$ is the symbol to be transmitted, and $t_k$ is the shrinking time, which is obtained from $t_k = T(2^{SF}-k)$.

Given that $\phi_d(t_0) = 0$, the phase shift caused by the Doppler shift can be given by $\phi_d(t) = \phi_d(t_0) + \int_{t_0}^{t} 2\pi f_d(\tau) \, d\tau$. Thus, we obtain the waveform function:
\begin{equation}
  \label{lora-con-wave}
  \begin{aligned} 
    s^\mathcal{C} _{k,d}(t) & =\begin{cases}
      \frac{1}{\sqrt{T_s}}e^{2j\pi \left(f_{\min}+\frac{B}{2^{SF}}(\frac{t-t_0}{2T} +k)\right)(t-t_0) \,+\, j\phi_d(t)},\\ 
      \hfill \text{if } 0\le t-t_0 \le t_k,\\
      \frac{1}{\sqrt{T_s}}e^{2j\pi \left(f_{\min}+\frac{B}{2^{SF}}(\frac{t-t_0}{2T} +k-2^{SF})\right)(t-t_0) \,+\,j\phi_d(t)},\\ 
      \hfill \text{if } t_k<  t-t_0 \le T_s. 
    \end{cases}
  \end{aligned}
\end{equation}

\textit{Discrete-time Domain:}
We sample the transmitted waveform with sampling time $T_d$. if $T_d = T$, there will be $2^{SF}$ samples. Otherwise, the number of samples will be $T_s/T_d$.  To make it general, we define $T_d = 2^s\cdot T$, then the frequency function is:
\begin{equation}
  \label{lora-dis-f}
  \begin{aligned} 
    f_{k,d}(n) & \!=\!\begin{cases}
      f_{\min}\!+\!f_d(n)\!+\!\frac{B\cdot2^s}{2^{SF}}(n\!-\!m_0\!+\!\frac{k}{2^s}),
      \\
      \hfill\text{if }0\!\le\! n\!-\!m_0 \!\le\! m_k,\\
      f_{\min}\!+\!f_d(n)\!+\!\frac{B\cdot2^s}{2^{SF}}(n\!-\!m_0\!+\!\frac{k-2^{SF}}{2^s}),
      \\
      \hfill\text{if }m_k \!<\!n\!-\!m_0 \!\le\! m_{\textit{SD}}. 
    \end{cases}
  \end{aligned}
\end{equation}
In the discrete time domain, the bandwidth $B=2^{\textit{s}}/T_d$, and $N$ is the number of samples, which can be calculated by $N = \left \lfloor \frac{2^{SF}}{2^s}  \right \rfloor$. The starting sample is $\textit{m}_0=t_0/T_d$, and the ending sample is $m_{\textit{SD}}=N-1$. In addition, the shrink sample is $m_k = \left \lfloor \frac{2^{SF}-k-1}{2^s}  \right \rfloor$. 

Similarly, with the sampled phase shift $\phi_d(n)$, the waveform function is 
\begin{equation}
  \label{lora-dis-wave}
  \begin{aligned} 
    s^{\mathcal{D}} _{k,d}(n) & =\begin{cases}
      \frac{1}{\sqrt{N}}e^{2j\pi \left(f_{\min}T_d+\frac{2^{2\textit{s}}}{2^{SF}}(\frac{n-m_0}{2} +\frac{k}{2^s})\right)(n-m_0)\,+\, j\phi_d(n)} , \\
      \hfill \text{ if }0\le \textit{n}-m_0 \le m_k,\\
      \frac{1}{\sqrt{N}}e^{2j\pi \left(f_{\min}T_d+\frac{2^{2s}}{2^{SF}}(\frac{n-m_0}{2}+\frac{k-2^{SF}}{2^s})\right)(n-m_0) \,+\, j\phi_d(n)}, \\
      \hfill \text{ if }m_k <  n-m_0 \le m_{\textit{SD}}.
    \end{cases}
  \end{aligned}
\end{equation}

\subsection{Satellite's Motion}
The satellite's motion geometry towards two ground devices is illustrated in Fig. \ref{orbit}, based on the model proposed in \cite{ali1998doppler}. We assume that the satellite has a prograde motion with a certain inclination angle, shown by the direction arrows of orbit in the figure. It should be noted that this model of satellite's motion is general and can be also applied to retrograde satellites or any inclination angle.

In the Earth Centered Earth Fixed (ECEF) coordinate system \cite{popescu2014pixel}, the satellite's orbit will not form a great circle due to Earth's self-rotation. But during the visibility window, we can assume the part of orbit as a great-circle arc, because the single visibility window in our scenario is very short compared to the orbit period, not to mention the shorter shared visibility window. For example, for a circular orbit with an altitude of $500$ km, the maximum visibility window duration is less than $13$ minutes, while the orbit period is about $1.57$ \rm{h} \cite{ali1998doppler}. More details about the visibility window will be introduced in the next section. 

In Fig. \ref{orbit}-(a), $O$ is the center of Earth, $A$ and $B$ are the stationary ground devices, $S$ is the satellite's location at time $t$, with $N$ as its subsatellite point, where the elevation angle of the ground device is $\theta(t)$. We define $S_{mj}$ ($j$ = $A$ or $B$) as the satellite's location where the corresponding ground device ($A$ or $B$) obtains its theoretical maximum elevation angle $\theta_{c\!j}$ ($j$ = $A$ or $B$), with $M_{j}$ ($j$ = $A$ or $B$) as its subsatellite point. We define $t_{c\!v\!j}$ ($j$ = $A$ or $B$) as the central visibility time when the satellite reaches location $S_{mj}$. Then we have $\theta(t_{c\!v\!A})=\theta_{c\!A}$ and $\theta(t_{c\!v\!B})=\theta_{c\!B}$. In addition, $\omega_S$ is the angular velocity of the satellite in Earth Centered Inertial (ECI) coordinate system \cite{popescu2014pixel}, $\omega_E$ is the Earth's self-rotation angular velocity, $i$ is the inclination angle which is the angle between satellite orbit plane and the equator plane, and $d$ is the distance between ground devices $A$ and $B$.

To make the following derivation clearer, we extract the main geometric relation of the satellite's motion towards the ground device and show it in Fig. \ref{orbit}-(b). This geometric relation is general as it can be applied to both $A$ and $B$, so we will not distinguish $A$ and $B$ anymore in the following derivation. To be concise, we use $G$ as the predefined $A$ or $B$, $M$ as the predefined $M_{j}$, $S_m$ as the predefined $S_{mj}$, $t_{cv}$ as the predefined $t_{c\!v\!j}$, and $\theta_c$ as the predefined $\theta_{c\!j}$. The derivations of the satellite's motion formula are given as follows.

As Fig. \ref{orbit}-(b) shows, the slant range $\rho(t)$, represented by $GS$, can be obtained by the following formula:
\begin{equation}
  \label{rho}
  \rho(t) = \sqrt{(R+H)^2+R^2-2R(R+H)\cos\gamma(t)},
\end{equation}
where $R$ is Earth's radius, $H$ is orbit altitude, and $\gamma(t)$ is the Earth's central angle between ground device and satellite. Based on the knowledge of trigonometric function, we can obtain the relation between the elevation angle and Earth's central angle, which is
\begin{equation}
  \label{gamma and elev}
  \cos\gamma(t) = \cos\left[\cos^{-1}\left (\frac{R}{R+H}\cos\theta(t)\right)-\theta(t)\right].
\end{equation}
Furthermore, it is important to notice that when the ground device gets its theoretical maximum elevation angle $\theta(t_{c\!v}) = \theta_c$, Plane $O\!G\!S_m$ will be perpendicular to Plane $O\!S\!S_m$\cite{ali1998doppler}, so the dihedral angle between them is $90$ degrees. Based on the cosine law of spherical triangle $M\!N\!G$, we get
\begin{equation}
  \label{gamma}
  \cos\gamma(t) = \cos\gamma(t_{c\!v})\cos\left(\varphi(t-t_{c\!v})\right),
\end{equation}
where $\varphi(t-t_{c\!v})$ represents the angle distance between $M$ and $N$ measured on the orbit plane. It is obtained by $\varphi(t-t_{c\!v}) = \omega_F(t-t_{c\!v})$. For a prograde satellite, considering the Earth's self-rotation, the angular velocity $\omega_F$ of the satellite in ECEF is
\begin{equation}
  \label{angular}
  \omega_F = \omega_S - \omega_E \cos i.
\end{equation}
The satellite angular velocity in ECI is expressed in the following lemma. After combining \eqref{rho} and \eqref{gamma}, we get the final expression for the slant range, whose derivative is the relative velocity shown in  \eqref{relative velocity}.

\begin{lemma}
The angular velocity of a satellite orbiting Earth in ECI at height \(H\) is:
\begin{equation}
    \label{Angular_velocity_ECI}
    \omega_S = \sqrt{\frac{gR^2}{(R + H)^3}},
\end{equation}
where \(g\) is gravitational acceleration, \(R\) is Earth's radius, and \(H\) is the satellite's orbit height.
\end{lemma}

\begin{proof}
Given the Newton's law of gravitation and centripetal force, we get
\begin{equation}
\frac{GMm}{(R+H)^2} =m\omega_S^2(R+H),\label{eqws}
\end{equation}
On the surface of the planet, the force between a mass \( m \) and a planet of mass \( M \) equals the gravitational force \(F = \frac{GMm}{R^2} = mg\). Therefore, we obtain $GM = gR^2$. By substituting it in \eqref{eqws}, we get \eqref{Angular_velocity_ECI}.

\end{proof}

\begin{figure*}
    \begin{equation}
    \label{relative velocity}
  v(t) = \dot{\rho} (t) = \frac{R(R+H)\sin\big(\varphi(t-t_{c\!v})\big)\big(\cos\gamma(t_{c\!v})\big) \dot{\varphi}(t)}{\sqrt{(R+H)^2+R^2-2R(R+H)\cos\big(\varphi(t-t_{c\!v})\big)\cos\gamma(t_{c\!v})}}.
    \end{equation}
    \begin{equation}
      \label{Doppler}
      f_d(t) = -\frac{f_k(t)}{c}\cdot\frac{R(R+H)\sin\big(\omega_F(t-t_{c\!v})\big)\cos\Big[\cos^{-1}\big (\frac{R}{R+H}\cos\theta_c\big)-\theta_c\Big]\omega_F}{\sqrt{(R+H)^2+R^2-2R(R+H)\cos\big(\omega_F(t-t_{c\!v})\big) \cos\Big[\cos^{-1}\big (\frac{R}{R+H}\cos\theta_c\big)-\theta_c\Big]}}.
    \end{equation}
    \hrulefill
\end{figure*}

\section{Doppler Shift Characterization} 
\label{Doppler_characteriztion}
In this section, we will derive the analytical expressions for the Doppler shift impacting the LoRa waveforms transmitted in our scenario. 

Indeed, the Doppler shift can be calculated as $f_d(t)=-f_k(t) \cdot(v(t)/c)$. Having \eqref{gamma and elev} and \eqref{relative velocity}, we can obtain the expression of the Doppler shift of the LoRa signal, as shown in \eqref{Doppler}. For simplicity, we only show the continuous-time formula here. From \eqref{Doppler}, we can see that the Doppler shift is a time-varying function, and it is impacted by the LoRa frequency function, orbit height, the satellite's angular velocity, and so on.

Meanwhile, we still need to find out when the satellite is visible to the ground device. In what follows, we will first define the single visibility window during which the satellite can communicate with one single given device. It is usually limited by elevation angle, satellite's velocity, etc. Then we will give the time range when the satellite can communicate with both two ground devices, namely the shared visibility window. We illustrate the single and shared visibility window for the ground devices $A$ and $B$. Then, we will introduce the Doppler linear approximation method which provides simplified and suitable expressions of Doppler shift for future cross-correlation derivation.

\subsection{Single Visibility Window}
As Fig. \ref{windows}-(a) shows, based on \eqref{gamma and elev}-\eqref{angular}, the single visibility window is expressed in the following lemma: 
\begin{lemma}\label{lemma1}
Due to the circular orbit, the single visibility window will be symmetric and is expressed as
\begin{equation}
  \label{single window}
  W = \big[t_{c\!v}-\Delta t_{\text{out}},t_{c\!v}-\Delta t_{\text{in}}\big]\cup \big[t_{c\!v}+\Delta t_{\text{in}},t_{c\!v}+\Delta t_{\text{out}}\big],
\end{equation}
where $\Delta t_{\text{out}}$ and $\Delta t_{\text{in}}$ are defined as
\begin{equation}
  \label{singlecalculation}
  \begin{aligned}
  &\Delta t_{\text{out/in}} = \\
  &\left | \frac{1}{\omega_F}\cos^{-1}\Big[\frac{\cos\Big(\cos^{-1}\big (\frac{R}{R+H}\cos\theta(t_{\text{out/in}})\big)-\theta(t_{\text{out/in}})\Big)}{\cos\Big(\cos^{-1}\big (\frac{R}{R+H}\cos\theta_c\big)-\theta_c\Big)}\Big] \right |.
  \end{aligned}
\end{equation}
Here $\theta(t_{\text{out}}) = \theta_{\min}$ and $\theta(t_{\text{in}}) = \theta_{\max}$ are the realistic minimum and maximum elevation angles respectively. 
\end{lemma}

\begin{proof}
It is worth noting that the theoretical maximum elevation angle $\theta_c$ may not be supported by the antenna or other physical structures of a ground device. There may be some realistic elevation angle limits for the ground device, and we use $[\theta_{\text{min}},\theta_{\text{max}}]$ to represent them. Meanwhile, let us denote $t_{\text{out}}=t_{cv} \pm \Delta t_{\text{out}}$ and $t_{\text{in}}=t_{cv} \pm \Delta t_{\text{in}}$ such that $\theta(t_{\text{out}}) = \theta_{\text{min}}$, and $\theta( t_{\text{in}}) = \theta_{\text{max}}$. Using \eqref{gamma and elev} and \eqref{gamma},  we can obtain \eqref{singlecalculation}.  
\end{proof}

\begin{remark}If the theoretical maximum elevation angle is higher than the realistic maximum elevation angle, the single visibility window will consist of two parts. If they are the same, the single visibility window will be one whole part. 
\end{remark}

\begin{figure}[ht]
  \centering
  \includegraphics[width=8.5cm]{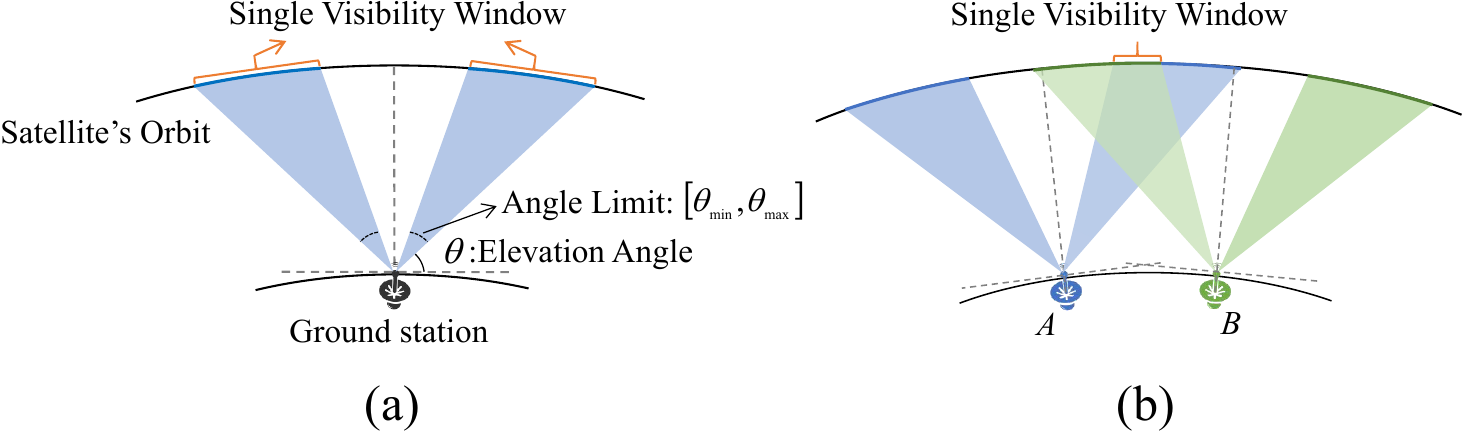}
  \caption{ Visibility window introduction: (a) Single visibility window;      (b) Shared visibility window.}
  \label{windows}
\end{figure}

Taking the ground device $A$ as the example, without loss of generality, we assume that $f_k(t) = f_c$, so that we can more clearly analyze the curves of Doppler shift and Doppler rate over the single visibility window, shown in Fig. \ref{shift_rate_window}. This figure is based on the following parameter setting: $i = 15^\circ$, $f_c = 868$ MHz, $B_1 = 250$ kHz, $H = 550$ km, the theoretical maximum elevation angle $\theta_{c\!A}=56^\circ$, and the realistic elevation angle limit $[\theta_{\min \!A}=10^\circ, \theta_{\max \!A}=50^\circ]$. 
\begin{figure}[ht]
  \centering
  \includegraphics[width=8cm]{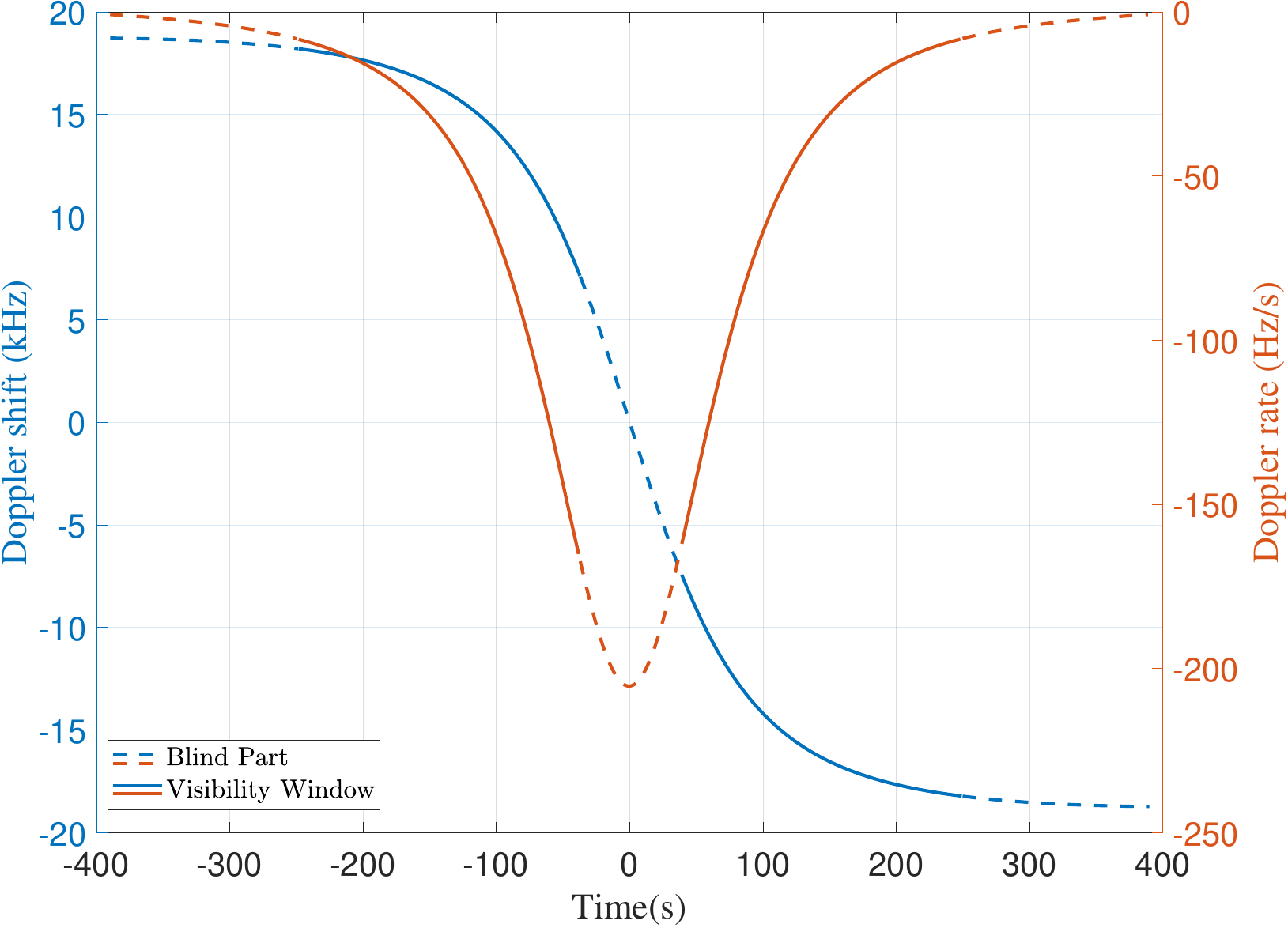}
  \caption{Simulation curve of Doppler shift and rate during the single visibility window, with carrier frequency 868 MHz and orbit height 550 km}
  \label{shift_rate_window}
\end{figure}

In Fig. \ref{shift_rate_window}, the solid lines represent the visibility window, and the dashed lines represent the blind part when no communication can be achieved between the ground device and the satellite. What's more, the blue and red lines represent the Doppler shift and Doppler rate respectively. The width of the visibility window (solid line) is about $422$ seconds, which is about $7$ minutes. Only considering the visibility window, for the Doppler shift, we can see that the absolute Doppler shift approximately ranges from $7.6$ kHz to $18.2$ kHz. For the Doppler rate, its absolute value ranges from $8.3$ Hz/s to $165.1$Hz/s. It should be noted that the Doppler shift cannot be $0$ Hz as the theoretical maximum elevation angle $\theta_{c\!A}=56^\circ$ cannot be achieved due to the realistic elevation angle limit. Meanwhile, based on Fig. \ref{shift_rate_window}, the transmission start time of the signal will determine the level of Doppler shift or Doppler rate it may experience.


\begin{remark} For Fig. \ref{shift_rate_window}, note that the theoretical maximum elevation angle determines how steep these two curves will be. The physical maximum and minimum elevation angles determine how wide the visibility window will be.
\end{remark}

\subsection{Shared Visibility Window}
When there are two ground devices, we need to express the shared visibility window for them, like Fig. \ref{windows}-(b) shows. To calculate it, we depict Fig. \ref{Twoplanes} to describe the geometry between two ground devices and one satellite. The blue ellipse denotes the Earth's equator plane, while the green ellipse denotes the orbit plane.
\begin{lemma}\label{lemma2}
Based on Lemma \ref{lemma1}, the shared visibility window is expressed as 
\begin{equation}
    \label{shared_window}
    W_{sh} = W_A \cap W_B,
\end{equation}
where $W_A$ and $W_B$ are the single visibility windows of devices $A$ and $B$, respectively.
\end{lemma}
\begin{proof}
    Applying Lemma \ref{lemma1}, we get the single visibility windows $W_A$ and $W_B$ for devices $A$ and $B$, respectively. The intersection of them will be the shared visibility window.
\end{proof}

\begin{figure}[ht]
  \centering
  \includegraphics[width=8.5cm]{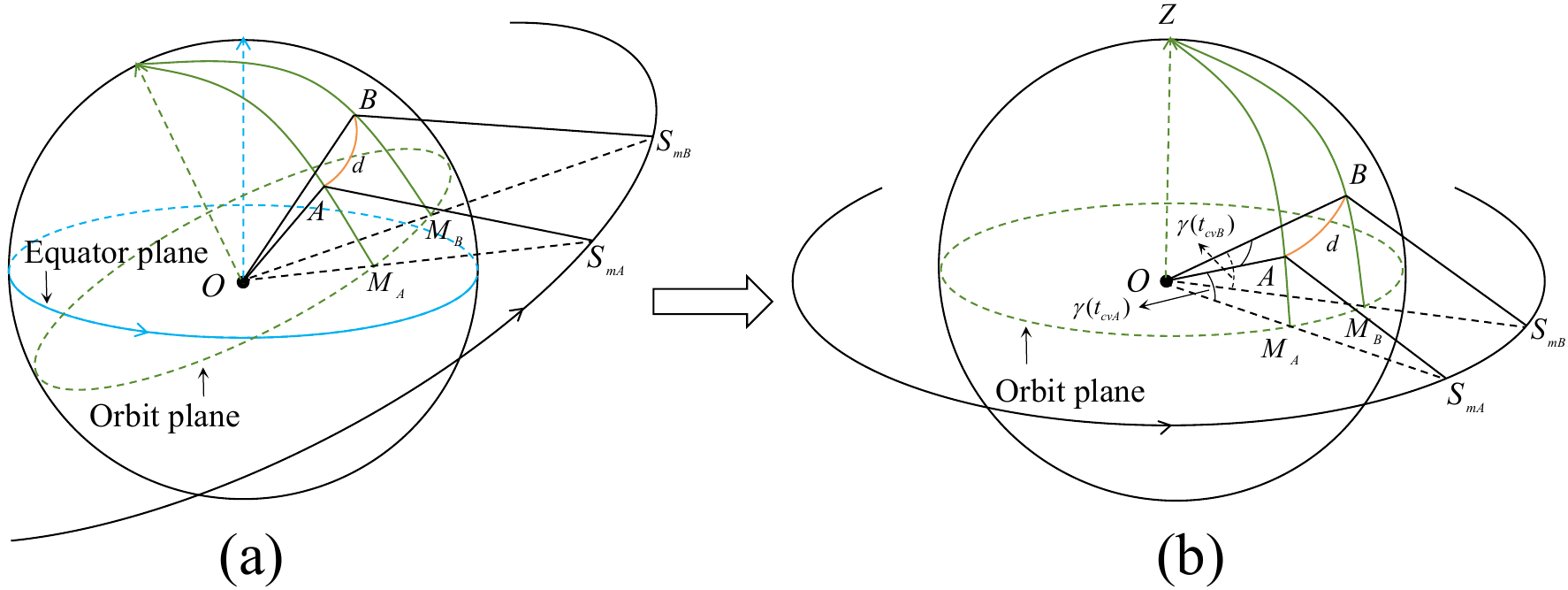}
  \caption{ Illustration about the base plane change: (a) Equator plane;      (b) Orbit plane.}
  \label{Twoplanes}
\end{figure}

In our scenario, about the relative location between two ground devices, we assume that we only know the theoretical and realistic elevation angle limits of them, as well as the distance $d$ between them. We assume the satellite passes $S_{m\!A}$ first, which means $t_{c\!v\!A}<t_{c\!v\!B}$, and we give $t_{c\!v\!B} = t_{c\!v\!A} + \Delta t_{A\!B}$. $\Delta t_{A\!B}$ is expressed in the following corollary:
\begin{corollary}
    \begin{equation}
  \label{time_distance}
  \Delta t_{A\!B} = \frac{1}{\omega_F}\cos^{-1}\Big[\frac{\cos(d/R)-\sin\gamma(t_{c\!v\!B})\sin\gamma(t_{c\!v\!A})}{\cos\gamma(t_{c\!v\!B})\cos\gamma(t_{c\!v\!A})}\Big].
\end{equation}
\end{corollary}
\begin{proof}
    If we rotate the whole sphere body of Fig. \ref{Twoplanes}-(a), letting the orbit plane be the base plane, we get Fig. \ref{Twoplanes}-(b). In this new sphere, the new north pole is denoted by $Z$, and the relative spherical distance between $A$ and $B$ is $d$, which is also the defined ground device distance, represented by the orange line. If we define the new longitude and latitude in this figure, like what we have for Earth, we will find the longitude is the angular distance that the satellite has moved with $\omega_F$, and the latitude can perfectly be denoted by $\gamma(t)$. Surprisingly, if the ground devices have the same theoretical maximum elevation angle, they will share the same absolute new latitude too. Therefore, for $A$ and $B$, their new latitude is $\gamma(t_{c\!v\!A})$ and $\gamma(t_{c\!v\!B})$ respectively, which can be directly calculated from \eqref{gamma and elev} using their theoretical maximum elevation angle. 
Then, in the body of $ZOAB$, we can calculate the final time distance from $\Delta t_{A\!B}=\angle POQ/\omega_F$. Here $\angle POQ$ is equal to the dihedral angle between Plane $AOZ$ and Plane $BOZ$. In addition, $\angle BOA = d/R$, $\angle BOZ = \pi/2-\gamma(t_{c\!v\!B})$, and $\angle AOZ = \pi/2-\gamma(t_{c\!v\!A})$. According to the cosine law of a spherical triangle, we have the following formula:
\begin{equation}
  \label{time_distance_calculation}
  \begin{aligned}
\cos\angle BOA= \cos\angle BOW \cos\angle AOW + \\
\hfill \sin\angle BOW \sin\angle AOW \cos\angle MON.
\end{aligned}
\end{equation}
From $\Delta t_{A\!B} = \angle BOA / \omega_F$, we can finally obtain \eqref{time_distance}.
\end{proof}

\subsection{Doppler Linear Approximation} \label{Doppler_linear_approximation}
As the derived formula of the Doppler shift in \eqref{Doppler} is too complicated to be used in the derivation of the cross-correlation function, we aim to simplify the Doppler shift function with a linear approximation method that was previously used in \cite{ben2022new,cao2021influence}. 
The approximation is applied for the Doppler shift within one symbol period. Thereby, with the nature of LoRa chirps \cite{benkhelifa2022orthogonal}, we can approximate the Doppler shift function by a piecewise linear function based on the original function in \eqref{Doppler}. Both continuous-time domain and discrete-time domain formulas are given below.

\textit{Continuous-time Domain:}
In the continuous-time domain, the linear Doppler shift function is defined as
\begin{equation}
  \label{approximation-con}
  \begin{aligned} 
    f_{dl}(t)  & =\begin{cases}
      c_d (t-t_0)+ v_d, &\text{ if $0\le t-t_0 \le t_k$ },\\
      c_d (t-t_0)+ v_d -\Delta f_k, & \text{ if $t_k < t-t_0 \le T_s$},
    \end{cases}
  \end{aligned}
\end{equation}
where $v_d$ is the initial Doppler shift defined as $v_d = f_d(t_0)$, $\Delta f_k = (f_{\min}+B)v(t_0+t_k)/c - f_{\min}v(t_0+t_k)/c = Bv(t_0+t_k)/c$, where $v(t_0+t_k)$ is the relative velocity at time instant $t_0+t_k$, and $c_d$ is the Doppler rate defined as $c_d = \big[f_d(t_0+T)-f_d(t_0)\big]/T$, which represents the slope of the Doppler shift curve.

\textit{Discrete-time Domain:}
Similarly, in the discrete-time domain, the linear Doppler shift function is defined as
\begin{equation}
  \label{approximation-dis}
  \begin{aligned} 
    f_{dl}(n)  & \!=\!\begin{cases}
      c_d (n\!-\!m_0)\!+\! v_d, &\text{if $0\!\le\! n\!-\!m_0 \!\le\! m_k$ },\\
      c_d (n\!-\!m_0)\!+\! v_d \!-\!\Delta f_k, & \text{if $m_k \!<\! n\!-\!m_0 \!\le\! m_{\textit{SD}}$},
    \end{cases}
  \end{aligned}
\end{equation}
where $v_d = f_d(m_0)$, $\Delta f_k = Bv(m_0+m_k)/c$, where $v(m_0+m_k)$ is the relative velocity at time instant $m_0+m_k$, $c_d = \Delta f_d / T_d$, and $\Delta f_d $ can be obtained from \eqref{Delta fd}
\begin{equation}
  \label{Delta fd}
  \begin{aligned} 
    \Delta f_d  & =\begin{cases}
      f_d(m_0+m_k)-f_d(m_0+m_k-1), \\
\hfill \text{ if $m_k+1 \geq m_{\textit{SD}}-m_k$ },\\
      f_d(m_0+m_k+2)-f_d(m_0+m_k+1),\\
\hfill \text{ if $m_k+1 < m_{\textit{SD}}-m_k$ }. 
    \end{cases}
  \end{aligned}
\end{equation}

To validate this type of approximation, we run some simulations which demonstrate the accurate estimation of the Doppler shift. The Doppler shift approximation matches perfectly with the analytical formula of Doppler shift in \eqref{Doppler}, and the mean absolute error (MAE) is very small, which is less than  $10^{-7}$ for the whole symbol period. 

Fig. \ref{approximation} shows an illustrative example of approximation results for one symbol period which demonstrates the accuracy of the approximation, where the blue line denotes the analytical Doppler shift, and the red circle denotes the approximated Doppler shift. 
Therefore, it is reasonable to use this linear function of the Doppler shift.
\begin{figure}[ht]
  \centering
  \includegraphics[width=8cm]{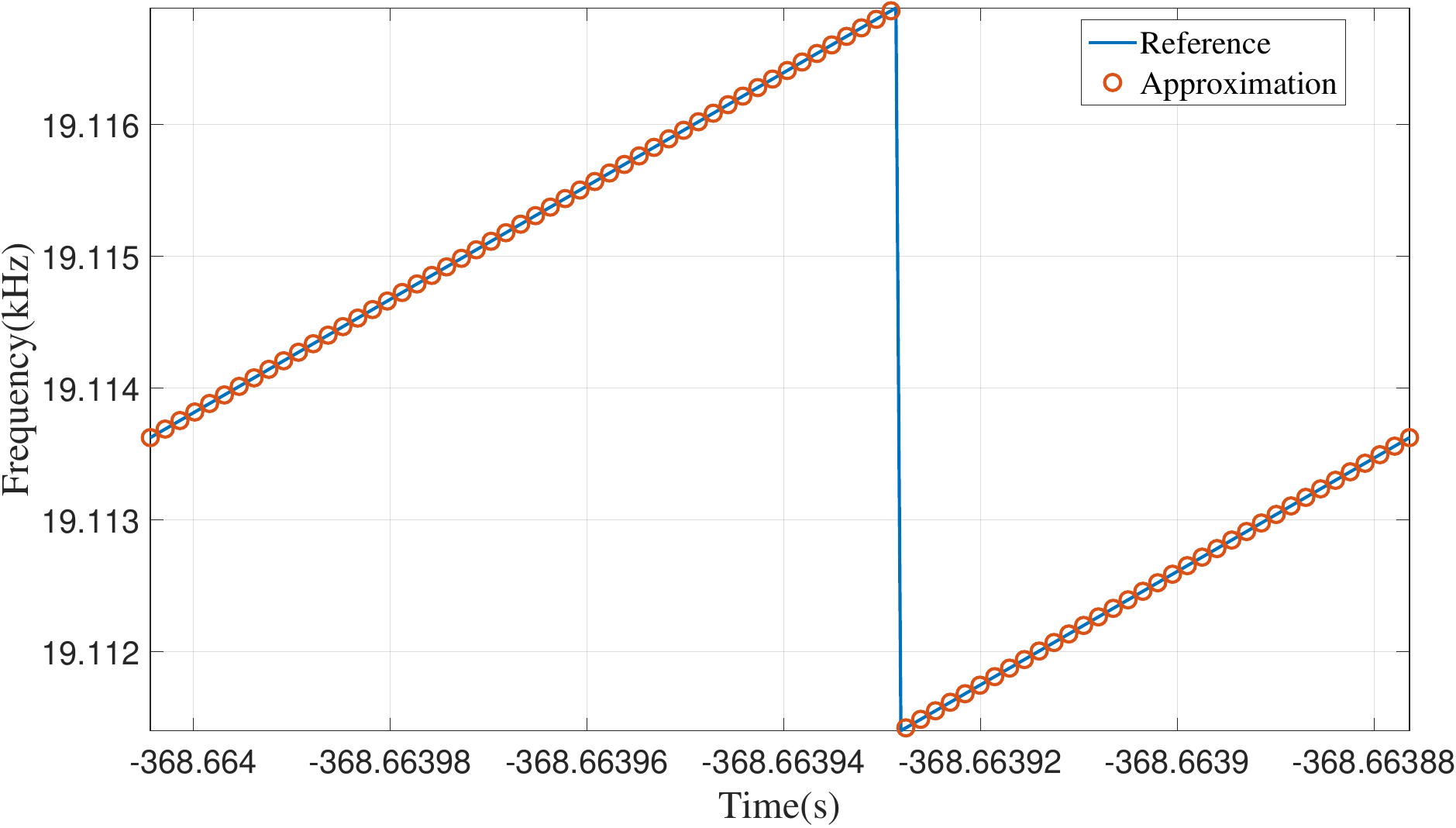}
  \caption{An illustrative example of the validation of the piecewise linear approximation for Doppler shift
  }
  \label{approximation}
\end{figure}
\section{Cross-Correlation Functions}\label{Cross-correlation_functions}
In this section, we will derive our analytical cross-correlation function between the signals received by a moving satellite from two ground devices $A$ and $B$ for both the continuous-time and discrete-time domains. Then some special cases for the cross-correlation function will be introduced. In addition, it is important to note that the following cross-correlation formulas are derived using the Doppler shift approximation formulas in Section \ref{Doppler_characteriztion}-\ref{Doppler_linear_approximation}.

For simplicity, in our following equations, all of the parameters belonging to ground device $A$ will be given subscript ‘$1$’, and those belonging to ground device $B$ will be given subscript ‘$2$’. For example, the bandwidth $B$ for ground device A is $B_1$, and for ground device B is $B_2$.
\vspace{-1pt}
\subsection{Continuous-Time Domain}
{To analyze the orthogonality of two LoRa waveforms in the presence of Doppler shift, by combining \eqref{lora-con-wave} and \eqref{approximation-con}, we firstly derive the final Doppler-impacted waveform expression in the continuous-time domain, shown in \eqref{final waveform-con}. }
\begin{equation}
  \label{final waveform-con}
  \begin{aligned} 
    &s^\mathcal{C} _{k,d}(t) =\\ &\begin{cases}
      \frac{1}{\sqrt{T_s}}e^{\left[2j\pi \left( f_{\text{min}}+v_d+\frac{B}{2^{SF}}(\frac{t-t_0}{2T} +k)+\frac{c_d(t-t_0)}{2} \right)(t-t_0)\right]},\\
      \hfill \text{ if $0\!\le\! t\!-\!t_0 \!\le\! t_k$ },\\
      \frac{1}{\sqrt{T_s}}e^{\left[2j\pi \big( f_{\text{min}}+v_d+\frac{B}{2^{SF}}(\frac{t-t_0}{2T} +k-2^{SF})+\frac{c_d(t-t_0)}{2} -\Delta f_k \big)(t-t_0)\right]},\\
      \hfill \text{ if $t_k\!<\!  t\!-\!t_0 \!\le\! T_s$}.
    \end{cases}
  \end{aligned}
\end{equation}
Then the cross-correlation function is defined as
\begin{equation}
  \label{cross-con}
  \mathcal{R}^\mathcal{C} _{k_1,k_2}(\tau,D_f,D_d)=\int_{-\infty}^{+\infty} s^\mathcal{C}_{k_1,d_1}(t+\tau)\Big(s^\mathcal{C}_{k_2,d_2}(t)\Big)^*\, dt,  
\end{equation}

where $s^\mathcal{C}_{k_1,d_1}(t)$ and $s^\mathcal{C}_{k_2,d_2}$ are the final waveform expressions for two LoRa signals, $\tau$ is the time delay, $D_f$ is the difference between the starting frequencies (DFS), i.e. $D_f = f_{\text{min}1}-f_{\text{min}2}$, and $D_d$ is the differential Doppler shift (DDS), and $D_d =v_{d_1} - v_{d_2}$.

As the waveform function in \eqref{final waveform-con} has two cases, we use $s^\mathcal{C} _{k,d}(t)[1]$ or $s^\mathcal{C} _{k,d}(t)[2]$ to denote the first or the second part of it respectively.
When doing the integral in that cross-correlation function, the integrand will have four types. 
For simplicity, we use $\mathcal{I}[i,j]$ to represent $s^\mathcal{C}_{k_1,d_1}(t+\tau)[i]\Big(s^\mathcal{C}_{k_2,d_2}(t)[j]\Big)^*$. From \eqref{cross-con}, we list all types of integrands below.
\begin{equation}
  \label{four types-con}
  \begin{aligned}
    \begin{cases}
      \mathcal{I}[1,1]
           \!=\! \frac{\mu}{\sqrt{T_{s1}T_{s2}}}e^{2j\pi\Big(y(t-t_0)+z\Big)(t-t_0)},\\
      \mathcal{I}[1,2]
           \!=\! \frac{\mu}{\sqrt{T_{s1}T_{s2}}}e^{2j\pi\Big(y(t-t_0)+z+\Delta f_{k2}+B_2\Big)(t-t_0)},\\
      \mathcal{I}[2,1]
           \!=\! \frac{\delta}{\sqrt{T_{s1}T_{s2}}}e^{2j\pi\Big(y(t-t_0)+z-\Delta f_{k1}-B_1\Big)(t-t_0)},\\
      \mathcal{I}[2,2]
           \!=\! \frac{\delta}{\sqrt{T_{s1}T_{s2}}}e^{2j\pi\Big(y(t-t_0)+z-\Delta f_{k1}+\Delta f_{k2}\!-\!B_1\!+\!B_2\Big)(t-t_0)}.
      \end{cases}      
  \end{aligned}
\end{equation}
where $y = \frac{1}{2}\Big(\frac{B_1}{2^{SF_1}T_1}-\frac{B_2}{2^{SF_2}T_2}\Big)  + \frac{c_{d1}-c_{d2}}{2}$, $f_j = f_{\text{min}j} + B_j  \frac{k_j}{2^{SF_j}} $, $\mu = \exp\Big[2j\pi \big(f_1+\frac{B_1}{ 2^{SF_1}}\frac{\tau}{2T_1}+\frac{c_{d1}\tau}{2}+v_{d1}\big)\tau\Big]$, $\delta = \exp\Big[2j\pi \big(f_1+\frac{B_1}{ 2^{SF_1}}(\frac{\tau}{2T_1}-2^{SF_1})+\frac{c_{d1}\tau}{2}+v_{d1}-\Delta f_{k_1}\big)\tau\Big]$, $z = D_f +D_d+h_{\tau}$, and $h_{\tau} = \tau\Big(\frac{B_1}{2^{SF_1}T_1}+c_{d1}\Big)$.

The two waveforms may have different mismatches in the time domain, resulting in various integral combinations in cross-correlation function. Let us denote $t_{\textit{a}} = t_0+\tau$, $t_{\textit{b}} = t_0+t_{k1}+\tau$, $t_{\textit{c}} = t_0+T_{s1}+\tau$, while $t_{\textit{x}} = t_0$, $t_{\textit{y}} = t_0+t_{k2}$, and $t_{\textit{z}} = t_0+T_{s2}$. Furthermore, we denote by $t_{\textit{be}}=\max\{t_{\textit{a}},t_{\textit{x}}\}$, and $t_{\textit{en}}=\min\{t_{\textit{c}},t_{\textit{z}}\}$. Subsequently, we have listed all of the possible cases for the final analytical expression of the cross-correlation function in Table \ref{Integral ranges}. 

\begin{table}[ht]
  \scriptsize
  \centering
  \caption{Cross-Correlation Functions in Continuous Time Domain}
  \label{Integral ranges}
  \begin{tabular}{lll}
  \toprule
    Type & Cross-correlation & Condition                     \\ \hline
    L1   & $\int_{t_{\textit{be}}}^{t_{\textit{en}}}\mathcal{I}[2,1]\,dt$                                                                 
    &  $t_b\!\le\! t_x\!\le\! t_c\!\le\! t_y$ \\
    L2   & $\int_{t_{\textit{be}}}^{t_{\textit{en}}}\mathcal{I}[1,2]\,dt$                                                                 
    &  $t_y\!\le\! t_a\!\le\! t_z\!\le\! t_b$ \\
    L3   & $\int_{t_{\textit{be}}}^{t_{b}}\mathcal{I}[1,1]\,dt\!+\!\int_{t_{b}}^{t_{\textit{en}}}\mathcal{I}[2,1]\,dt$                                 &  $t_x\!\le\! t_b\!\le\! t_c\!\le\! t_y$ \\
    L4   & $\int_{t_{\textit{be}}}^{t_{b}}\mathcal{I}[1,2]\,dt\!+\!\int_{t_{b}}^{t_{\textit{en}}}\mathcal{I}[2,2]\,dt$                                 &  $t_y\!\le\! t_a\!\le\! t_b\!\le\! t_z$ \\
    L5   & $\int_{t_{\textit{be}}}^{t_{y}}\mathcal{I}[1,1]\,dt\!+\!\int_{t_{y}}^{t_{\textit{en}}}\mathcal{I}[1,2]\,dt$                                 &  $t_a\!\le\! t_y\!\le\! t_z\!\le\! t_b$ \\
    L6   & $\int_{t_{\textit{be}}}^{t_{y}}\mathcal{I}[2,1]\,dt\!+\!\int_{t_{y}}^{t_{\textit{en}}}\mathcal{I}[2,2]\,dt$                                 &  $t_b\!\le\! t_x\!\le\! t_y\!\le\! t_c$ \\
    L7   & $\int_{t_{\textit{be}}}^{t_{b}}\mathcal{I}[1,1]\,dt\!+\!\int_{t_{b}}^{t_{y}}\mathcal{I}[2,1]\,dt\!+\!\int_{t_{y}}^{t_{\textit{en}}}\mathcal{I}[2,2]\,dt$ & $t_x\!\le\! t_b\!\le\! t_y\!\le\! t_c$ \\
    L8   & $\int_{t_{\textit{be}}}^{t_{y}}\mathcal{I}[1,1]\,dt\!+\!\int_{t_{y}}^{t_{b}}\mathcal{I}[1,2]\,dt\!+\!\int_{t_{b}}^{t_{\textit{en}}}\mathcal{I}[2,2]\,dt$ & $t_a\!\le\! t_y\!\le\! t_b\!\le\! t_z$ \\ 
  \bottomrule
  \end{tabular}
\end{table}
\begin{table}[ht]
  \centering
  \scriptsize
  \caption{Cross-Correlation Functions in Discrete Time Domain}
  \label{Summation ranges}
  \begin{tabular}{lll}
  \toprule
    Type & Cross-correlation & Condition                     \\ \hline
    L1   & $\sum_{m_{\textit{be}}}^{m_{\textit{en}}}\mathcal{S}[2,1]$                                                                 
    &  $m_{\textit{b}}\!\le\! m_{\textit{x}}\!\le\! m_{\textit{c}}\!\le\! m_{\textit{y}}$ \\
    L2   & $\sum_{m_{\textit{be}}}^{m_{\textit{en}}}\mathcal{S}[1,2]$                                                                 
    &  $m_{\textit{y}}\!\le\! m_{\textit{a}}\!\le\! m_{\textit{z}}\!\le\! m_{\textit{b}}$ \\
    L3   & $\sum_{m_{\textit{be}}}^{m_{\textit{b}}-1}\mathcal{S}[1,1]\!+\!\sum_{m_{\textit{b}}}^{m_{\textit{en}}}\mathcal{S}[2,1]$                                 &  $m_{\textit{x}}\!\le\! m_{\textit{b}}\!\le\! m_{\textit{c}}\!\le\! m_{\textit{y}}$ \\
    L4   & $\sum_{m_{\textit{be}}}^{m_{\textit{b}}-1}\mathcal{S}[1,2]\!+\!\sum_{m_{\textit{b}}}^{m_{\textit{en}}}\mathcal{S}[2,2]$                                 &  $m_{\textit{y}}\!\le\! m_{\textit{a}}\!\le\! m_{\textit{b}}\!\le\! m_{\textit{z}}$ \\
    L5   & $\sum_{m_{\textit{be}}}^{m_{\textit{y}}-1}\mathcal{S}[1,1]\!+\!\sum_{m_{\textit{y}}}^{m_{\textit{en}}}\mathcal{S}[1,2]$                                 &  $m_{\textit{a}}\!\le\! m_{\textit{y}}\!\le\! m_{\textit{z}}\!\le\! m_{\textit{b}}$ \\
    L6   & $\sum_{m_{\textit{be}}}^{m_{\textit{y}}-1}\mathcal{S}[2,1]\!+\!\sum_{m_{\textit{y}}}^{m_{\textit{en}}}\mathcal{S}[2,2]$                                 &  $m_{\textit{b}}\!\le\! m_{\textit{x}}\!\le\! m_{\textit{y}}\!\le\! m_{\textit{c}}$ \\
    L7   & $\sum_{m_{\textit{be}}}^{m_{\textit{b}}-1}\mathcal{S}[1,1]\!+\!\sum_{m_{\textit{b}}}^{m_{\textit{y}}}\mathcal{S}[2,1]\!+\!\sum_{m_{\textit{y}}+1}^{m_{\textit{en}}}\mathcal{S}[2,2]$ &  $m_{\textit{x}}\!\le\! m_{\textit{b}}\!\le\! m_{\textit{y}}\!\le\! m_{\textit{c}}$ \\
    L8   & $\sum_{m_{\textit{be}}}^{m_{\textit{y}}-1}\mathcal{S}[1,1]\!+\!\sum_{m_{\textit{y}}}^{m_{\textit{b}}}\mathcal{S}[1,2]\!+\!\sum_{m_{\textit{b}}+1}^{m_{\textit{en}}}\mathcal{S}[2,2]$ &  $m_{\textit{a}}\!\le\! m_{\textit{y}}\!\le\! m_{\textit{b}}\!\le\! m_{\textit{z}}$ \\ 
  \bottomrule
  \end{tabular}
\end{table}

\begin{remark}
If the time delay $\tau=0$, then $\mu = \delta = 1$, and $h_\tau=0$.
If they use the same carrier frequency, bandwidth, as well as the $SF$, i.e. $f_{c1} = f_{c2}$, $B_1 = B_2$, and $SF_1 = SF_2$, then $D_f=0$, $y = \frac{c_{d1}-c_{d2}}{2}$, and $z = D_d$. 
\end{remark}


\subsection{Discrete-Time Domain}
Similarly, by combining \eqref{lora-dis-wave} and \eqref{approximation-dis}, we derive the final Doppler-impacted waveform expression in the discrete-time domain, shown in \eqref{final waveform-dis}. 
\begin{equation}
  \label{final waveform-dis}
\begin{aligned}
&s^\mathcal{D} _{k,d}(n) \!=\!\\ &\left\{
\begin{array}{*{2}{ll}}
\frac{1}{\sqrt{N}}\exp \left. \Big[ 2j\pi\big(f_{\min}T_d\!+\!v_dT_d\!+\!\frac{2^{2s}}{2^{SF}}(\frac{n-m_0}{2} \!+\!\frac{k}{2^s} )\right.\\ 
\left. +\frac{c_dT_d^2(n-m_0)}{2}\big)(n\!-\!m_0)\Big]\right.,  \hfill \text{if $0\!\le\! n\!-\!m_0 \!\le\! m_k$},\\
\frac{1}{\sqrt{N}}\exp \left. \Big[ 2j\pi\big(f_{\min}T_d\!+\!v_dT_d\!+\!\frac{2^{2s}}{2^{SF}}(\frac{n-m_0}{2} \!+\!\frac{k-2^{SF}}{2^s} )\right.\\ 
\left. +\frac{c_dT_d^2(n-m_0)}{2}\!-\!\Delta f_kT_d\big)(n\!-\!m_0)\Big]\right., \hfill \text{if $m_k \!<\!  n\!-\!m_0 \!\le\! m_{\textit{SD}}$}.\\
\end{array} 
\right.
\end{aligned}
\end{equation}
Then the cross-correlation function is 
\begin{equation}
  \label{cross-dis}
  \mathcal{R}^\mathcal{D} _{k_1,k_2}(m_\tau,D_f,D_d)=\sum_{-\infty }^{+\infty} s^\mathcal{D}_{k_1,d_1}(n+m_\tau )\Big(s^\mathcal{D}_{k_2,d_2}(n)\Big)^*,  
\end{equation}
where $m_\tau = \tau/T_d$, and the other parameters are the same as defined in the continuous-time domain.

As the waveform function has two cases, we use $s^\mathcal{D} _{k,d}(n)[1]$ or $s^\mathcal{D} _{k,d}(n)[2]$ to denote the first or the second part of it respectively. 
When doing the summation in that cross-correlation function, the result may consist of different summations. We have listed all possible types of equations below. We use $\mathcal{S}[i,j]$ to represent $s^\mathcal{D}_{k_1,d_1}(n+m_\tau)[i]\Big(s^\mathcal{D}_{k_2,d_2}(n)[j]\Big)^*$. From \eqref{cross-dis}, the analytical expressions will be 
\begin{equation}
  \label{four types}
  \begin{aligned}
    \begin{cases}
      \mathcal{S}[1,1]
           = \frac{\xi}{\sqrt{N_1 N_2}}e^{2j\pi\Big(y(n-m_0)+z\Big)(n-m_0)},\\
      \mathcal{S}[1,2]
           = \frac{\xi}{\sqrt{N_1 N_2}}e^{2j\pi\Big(y(n-m_0)+z+2^{s2}+\Delta f_{k2}T_d\Big)(n-m_0)},\\
      \mathcal{S}[2,1]
           = \frac{\gamma}{\sqrt{N_1 N_2}}e^{2j\pi\Big(y(n-m_0)+z-2^{s1}-\Delta f_{k1}T_d\Big)(n-m_0)},\\
      \mathcal{S}[2,2] = \\
            \frac{\gamma}{\sqrt{N_1 N_2}}e^{2j\pi\Big(y(n-m_0)+z-2^{s1}+2^{s2}-\Delta f_{k1}T_d+\Delta f_{k2}T_d\Big)(n-m_0)}.\\
      \end{cases}      
  \end{aligned}
\end{equation}
where $y \!=\! \frac{1}{2}\Big(\frac{2^{2s_1}}{2^{SF_1}}-\frac{2^{2s_2}}{2^{SF_2}}\Big) +\frac{(c_{d1}-c_{d2})T_d^2}{2}$, $\xi \!=\! \exp\Big[2j\pi \big(f_{\min1}T_d+\frac{2^{2s_1}}{ 2^{SF_1}}(\frac{m_{\tau}}{2}+\frac{k_1}{2^{s_1}})+\frac{c_{d_1}m_{\tau}T_d^2}{2}+v_{d_1}T_d\big)m_\tau\Big]$, $\gamma = \exp\Big[2j\pi \big(f_{\min1}T_d+\frac{2^{2s_1}}{ 2^{SF_1}}(\frac{m_{\tau}}{2}+\frac{k_1-2^{SF}}{2^{s_1}})+\frac{c_{d_1}m_{\tau}T_d^2}{2}+v_{d_1}T_d-\Delta f_{k_1}T_d\big) m_\tau\Big]$, $z = D_f T_d  +D_d T_d+h_{\tau}$, and $h_{\tau} = m_\tau(\frac{2^{2s_1}}{2^{SF_1}} +c_{d_1} T_d^2)$.

We define $m_{\textit{a}} = m_0+m_\tau$, $m_{\textit{b}} = m_0+m_{k1}+m_\tau$, $m_{\textit{c}} = m_0+m_{\textit{SD}1}+m_\tau$, while $m_{\textit{x}} = m_0$, $m_{\textit{y}} = m_0+m_{k2}$, and $m_{\textit{z}} = m_0+m_{\textit{SD}2}$. Then we let $m_{\textit{be}}=\max\{m_{\textit{a}},m_{\textit{x}}\}$, and $m_{\textit{en}} = \min\{m_{\textit{c}},m_{\textit{z}}\}$. Similarly, we have also listed all of the possible cases for the final analytical expression of the cross-correlation function in the discrete case in Table \ref{Summation ranges}. 

\begin{remark}
If the time delay $m_\tau=0$, then $\xi =1$, $\gamma = 1$, and $h_\tau=0$. Furthermore, assuming they use the same carrier frequency $f_c$, if $B_1 = B_2$, $s_1 = s_2$, and $SF_1 = SF_2$, then $D_f=0$, $y = \frac{(c_{d1}-c_{d2})T_d^2}{2}$, and $z = D_dT_d$.
\end{remark}

\section{Numerical Results} \label{Numerical_results}
In this section, we aim to characterize the cross-correlation factors in both continuous and discrete time domains, focusing on mean and maximum cross-correlations across all possible symbols. 

What's more, we will further compare these results between the no-Doppler and with-Doppler cases. Specifically, for the case with Doppler effects, we will conduct a comprehensive analysis of the impact of the Doppler shift on both the LoRa signal and cross-correlation.

\subsection{Parameter Setting}\label{para_setting}
Firstly, we establish the following parameters: the speed of light $c = 3\times 10^8\,\rm{m/s}$, gravitational acceleration $g = 9.8\,\rm{m/s^2}$, Earth radius $R = 6371\,\rm{km}$, and Earth's self-rotation angular velocity $\omega_E = 7.292\times 10^{-5}\,\rm{rad/s}$.

Additionally, unless explicitly stated otherwise, we consider the following scenario: the spherical distance between two ground devices $d = 10$ km, the inclination angle of the orbit $i = 15^\circ$, carrier frequency $f_c = 868$ MHz, bandwidth $B_1 = B_2 = 250$ kHz, and orbit height $H = 550$ km. We set the theoretical maximum elevation angle $\theta_{c\!A}=56^\circ$ and $\theta_{c\!B}=56.2^\circ$, along with the realistic elevation angle limits $[\theta_{\min \!A}=10^\circ, \theta_{\max \!A}=50^\circ]$ and $[\theta_{\min \!B}=10^\circ, \theta_{\max \!B}=50^\circ]$. Moreover, we assume that both ground devices utilize the same carrier frequency. 

\subsection{No-Doppler Case}
Firstly, we investigate the no-Doppler case, where $c_d$, $v_d$, and $\Delta f_k$ are all set to 0, consistent with the simulation approach of \cite{benkhelifa2022orthogonal}. However, we present only the distinct results here\footnote{While processing the complex cross-correlation values, we initially extract either the mean or maximum value and subsequently obtain the absolute value in the final step for both continuous and discrete time domains. However, in \cite{benkhelifa2022orthogonal}, their processing sequence differs between the continuous and discrete time domains, leading to noticeable inconsistency in results across the two domains.}.

\subsubsection{Same Bandwidth}
In the Fig. \ref{Corr_no_Doppler_Con_Dis}, with $\tau = 0$, $B_1=B_2=250$ kHz, we plotted the maximum and mean cross-correlation functions in both continuous-time and discrete-time domain.
We have validated the correctness of the derived analytical expressions in Table \ref{Integral ranges} and Table \ref{Summation ranges} by ensuring that they produce the same output results as those from the defining formulas in \eqref{cross-con} and \eqref{cross-dis}, which can be applied to the results of this figure and all of the following figures.

In Figs. \ref{Corr_no_Doppler_Con_Dis}-(a) and \ref{Corr_no_Doppler_Con_Dis}-(b), for the maximum cross-correlation, we notice the similar features of both time domains: 
\begin{itemize}
    \item The maximum cross-correlation is close to 1 when $SF_1 = SF_2$, significantly surpassing the case when $SF_1 \neq SF_2$. This suggests that LoRa signals are unlikely to be orthogonal when $SF_1 = SF_2$.
    \item When $SF_1 \neq SF_2$, lower $SF_2$ (i.e., $SF_2<SF_1$) exhibit higher maximum cross-correlation compared to higher $SF_2$ (i.e., $SF_2>SF_1$). Furthermore, as the difference between $SF$s decreases, the maximum cross-correlation tends to increase.
    \item If we exchange those two $SF$s, the maximum cross-correlation will stay the same, and this symmetric property can easily be verified from the analytical expressions.
\end{itemize}

Comparing the maximum cross-correlation factors between continuous and discrete time domains, when $SF_1 \neq SF_2$, we can see the values in the continuous-time domain are a little bit higher than those in the discrete-time domain.
\begin{figure}[ht]
  \centering
  \includegraphics[width=8.4cm]{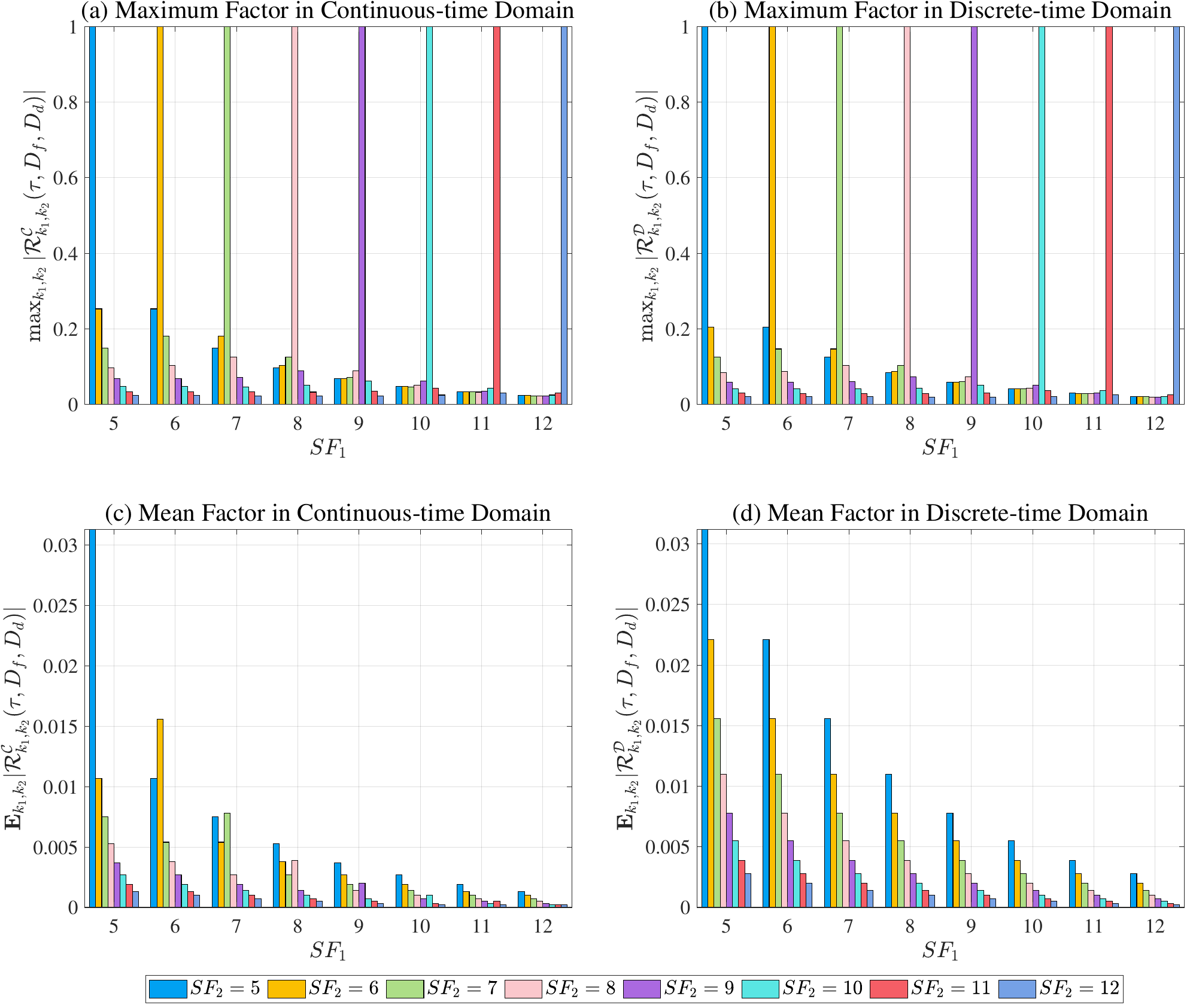}
  \caption{Maximum and mean cross-correlation factor between $SF$s without Doppler effect in both continuous-time and discrete-time domain.}
  \label{Corr_no_Doppler_Con_Dis}
\end{figure}
In Figs. \ref{Corr_no_Doppler_Con_Dis}-(c) and \ref{Corr_no_Doppler_Con_Dis}-(d), regarding the mean cross-correlation, we can notice different behaviors from Figs. \ref{Corr_no_Doppler_Con_Dis}-(a) and \ref{Corr_no_Doppler_Con_Dis}-(b):
\begin{itemize}
    \item $SF_1 = SF_2$ will no longer always hold the highest value for mean cross-correlation.
    \item When the difference between $SF$s increases, the value may increase. 
    \item When $SF_1 = SF_2$, the mean cross-correlations in both time domains are the same. When $SF_1 \neq SF_2$, the mean values in the discrete-time domain will be higher than those in the continuous time domain.
\end{itemize}

During the simulation for the two figures above, we have testified that as long as $B_1 = B_2$ and $\tau=0$, the realistic value of bandwidth has no impact on the final cross-correlation, which can also be proved from the analytical expressions.

\subsubsection{Different Bandwidth} 
In this part, we display the maximum cross-correlation results under different bandwidths, i.e. $B_1 \neq B_2$. We will only report the different results from \cite{benkhelifa2022orthogonal}.
Without losing generality, we consider the $SF$ ranging from $5$ to $10$ under both continuous and discrete time domains. The results are shown in Fig. \ref{no_doppler_con_dis_BW}.

\begin{figure}[ht]
  \centering
  \includegraphics[width=8.4cm]{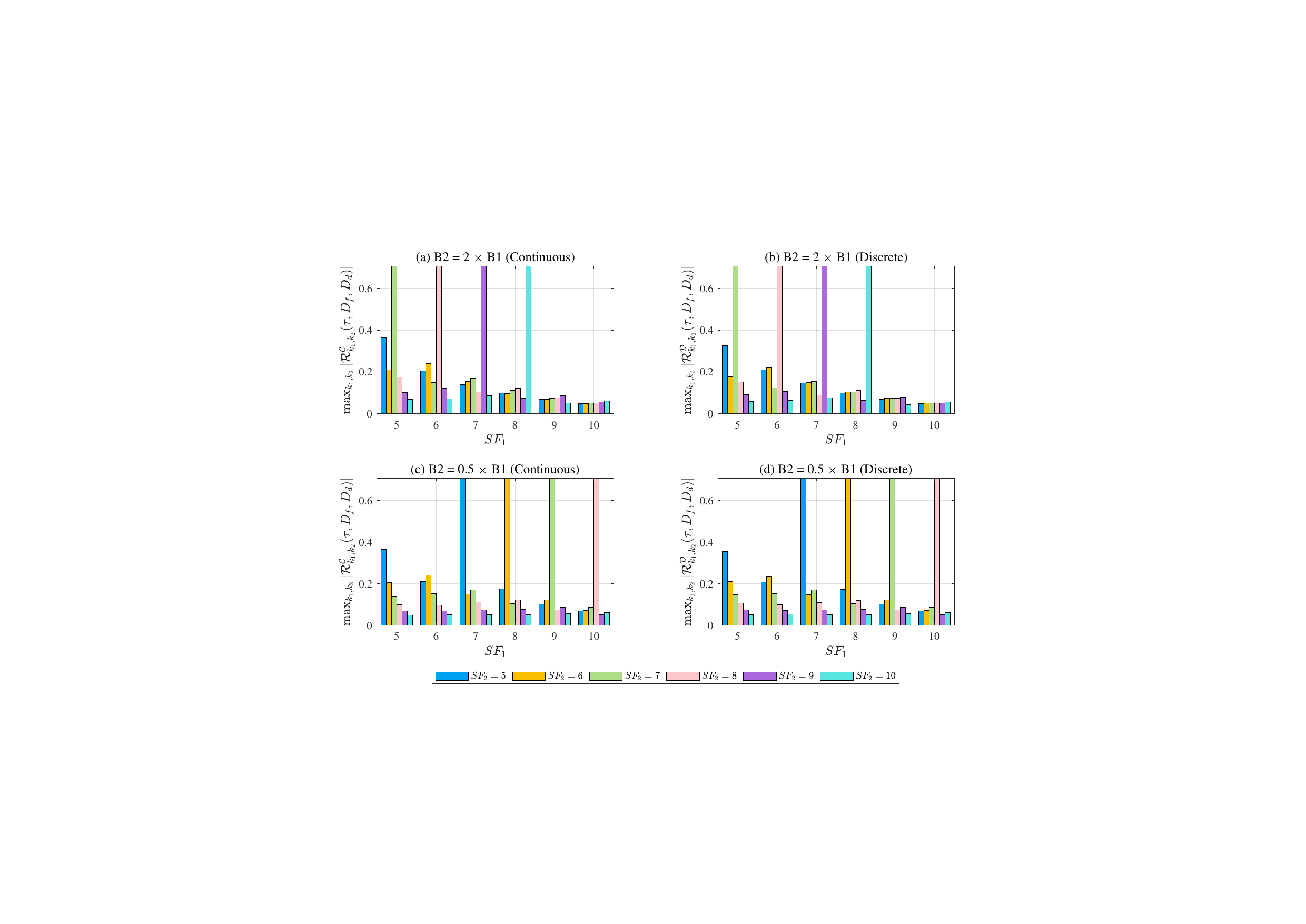}
  \caption{Maximum cross-correlation factor between $SF$s without Doppler effect given different bandwidths in both continuous and discrete time domains.}
  \label{no_doppler_con_dis_BW}
\end{figure}
We can observe that continuous and discrete time domains have nearly the same behaviors. In \cite{pham2020radio}, it introduced the chirp rate $C_r = B^2/2^{SF}$. Then in Fig. \ref{no_doppler_con_dis_BW}, we can conclude that the highest value of the maximum cross-correlation takes place when $C_{r1} = C_{r2}$. We can also observe that the highest value decreases to 0.707 in the cases when $B_1 = 2\times B_2$ or $B_1  = 0.5 \times B_2$. 

\subsection{With-Doppler Case}
In this part, we will illustrate the impact of the Doppler effect on the final cross-correlation or the orthogonality of LoRa signals. We have plotted the maximum and mean cross-correlation results in both time domains with Doppler shift and using the same bandwidth $B_1=B_2$ and $\tau=0$, which are shown in Figs. \ref{with_doppler_con} and \ref{with_doppler_dis}.  

In the continuous-time domain, for the mean cross-correlation factor, in Figs. \ref{with_doppler_con}-(b) and \ref{with_doppler_con}-(d), it remains relatively constant for both high Doppler shift and high Doppler rate compared to the no-Doppler case. 
For the maximum cross-correlation factor, when $SF_1 \neq SF_2$, the maximum cross-correlation remains largely unaffected by both high Doppler shift or high Doppler rate. In Fig. \ref{with_doppler_con}-(a), with $SF_1 = SF_2$ and high Doppler shift, it exhibits a slight decrease only when $SF$ reaches 11 or 12, particularly when $SF_1 = SF_2$. However, in Fig. \ref{with_doppler_con}-(c), with high Doppler rate and $SF_1 = SF_2$, the maximum cross-correlation factors experience obvious decreases while $SF_1$ reaches $8$, $9$, and $12$. 
\begin{figure}[ht]
  \centering
  \includegraphics[width=8.5cm]{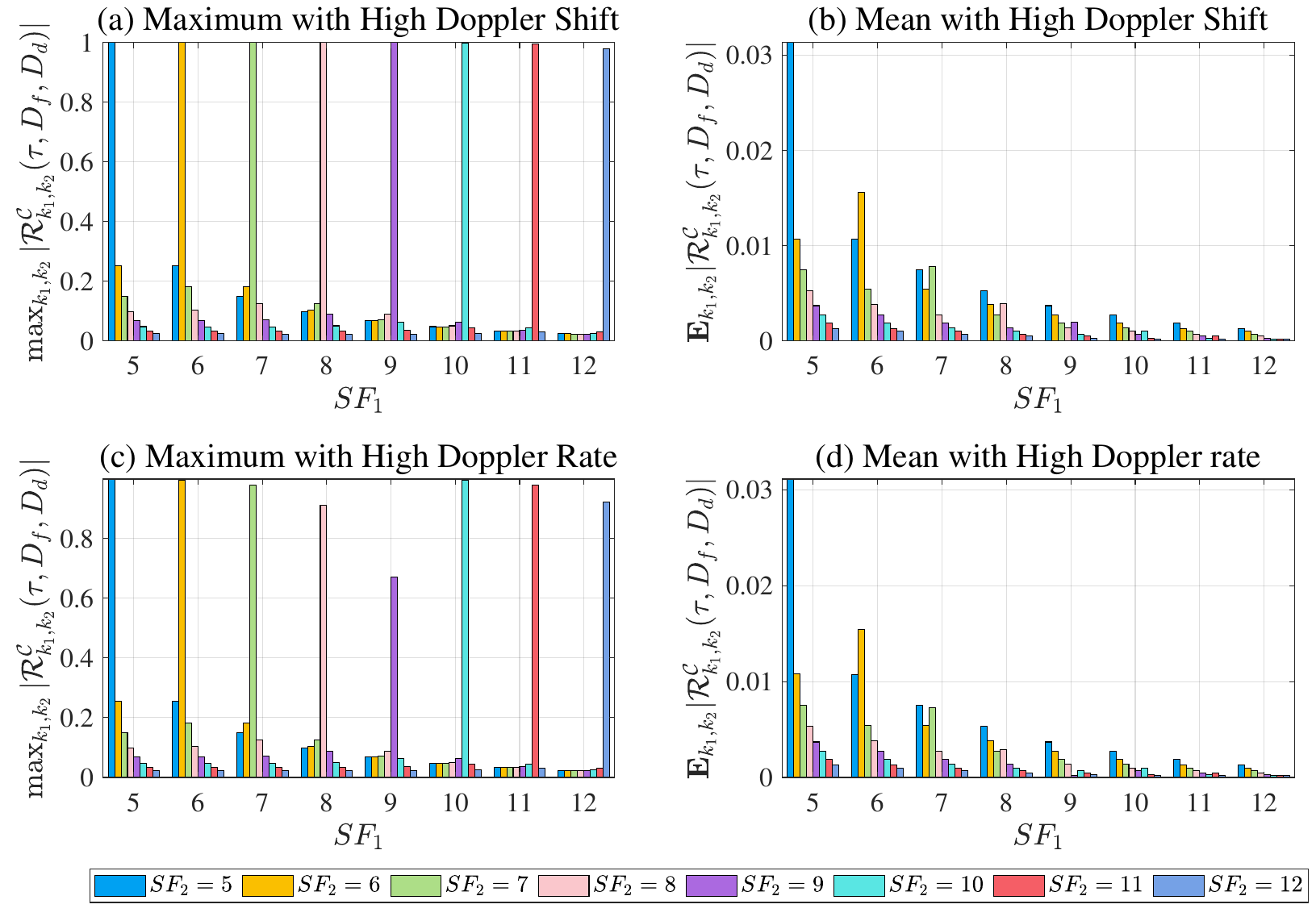}
  \caption{Maximum and Mean cross-correlation factors between $SF$s under high Doppler shift or Doppler rate in continuous-time domain.}
  \label{with_doppler_con}
\end{figure}
\begin{figure}[ht]
  \centering
  \includegraphics[width=8.5cm]{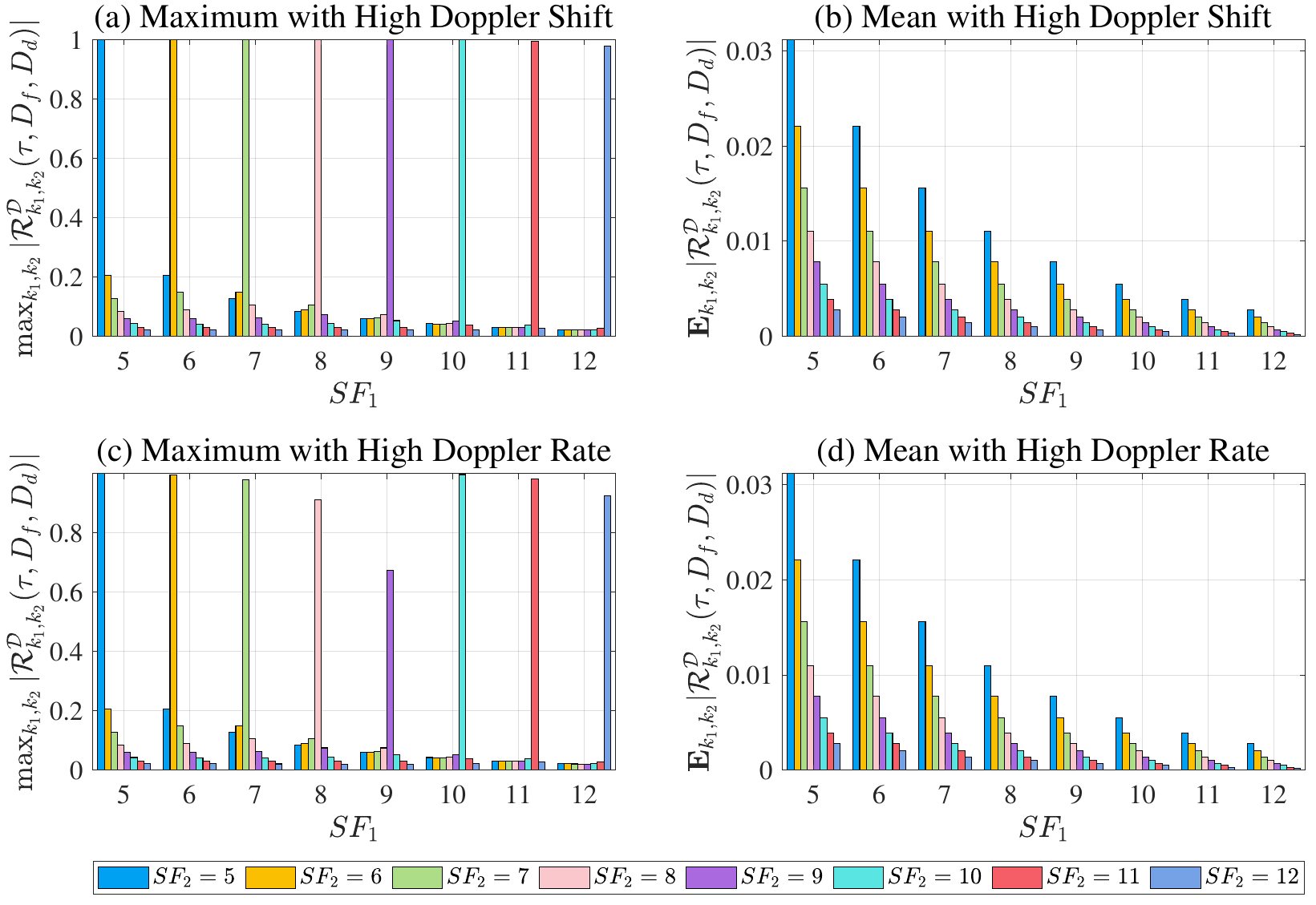}
  \caption{Maximum and Mean cross-correlation factors between $SF$s under high Doppler shift or Doppler rate in discrete-time domain.}
  \label{with_doppler_dis}
\end{figure}

In the discrete-time domain, the mean and maximum cross-correlation factors exhibit the same behaviors observed in the continuous-time domain, shown in Fig. \ref{with_doppler_dis}.

Therefore, for both continuous-time and discrete-time domains, we can conclude that the maximum cross-correlation with $SF_1 \neq SF_2$ and the mean cross-correlation is immune to the Doppler effect. However, the maximum cross-correlation with $SF_1 = SF_2$ is only immune to the high Doppler shift and not immune to the high Doppler rate.

\subsection{Explanation of the Impact of Doppler Shift}
In this subsection, we aim to provide a theoretical explanation strategy for the noticeable decrease observed in the preceding cross-correlation results under high Doppler rate. We will establish the relationship from transmission start time to differential Doppler shift $D_d$, to carrier frequency difference, and finally to the cross-correlation value. Our following analysis is based on the discrete-time domain and can be equivalently applicable to the continuous-time domain. This comprehensive analysis will enable us to interpret the impact of the Doppler shift on Figs. \ref{with_doppler_con} and \ref{with_doppler_dis}.

Firstly, for the relation between transmission start time and $D_d$, we can plot its curve based on $D_d = v_{d1} - v_{d2}$ and \eqref{Doppler}, shown in Fig. \ref{Dd_over_window}. Without losing generality, in this figure, we only show the $D_d$ values over the first part of the shared visibility window, which spans from transmission start time with a high Doppler shift to that with a high Doppler rate. It deserves noting that for any $SF_1$ or $SF_2$, the $D_d$ has the same variation as Fig. \ref{Dd_over_window}, approximately ranging from $-6 $ Hz and $-230$ Hz. This indicates that the high Doppler rate will cause a higher $|D_d|$ than the high Doppler shift. 

\begin{figure}[ht]
  \centering
  \includegraphics[width=8cm]{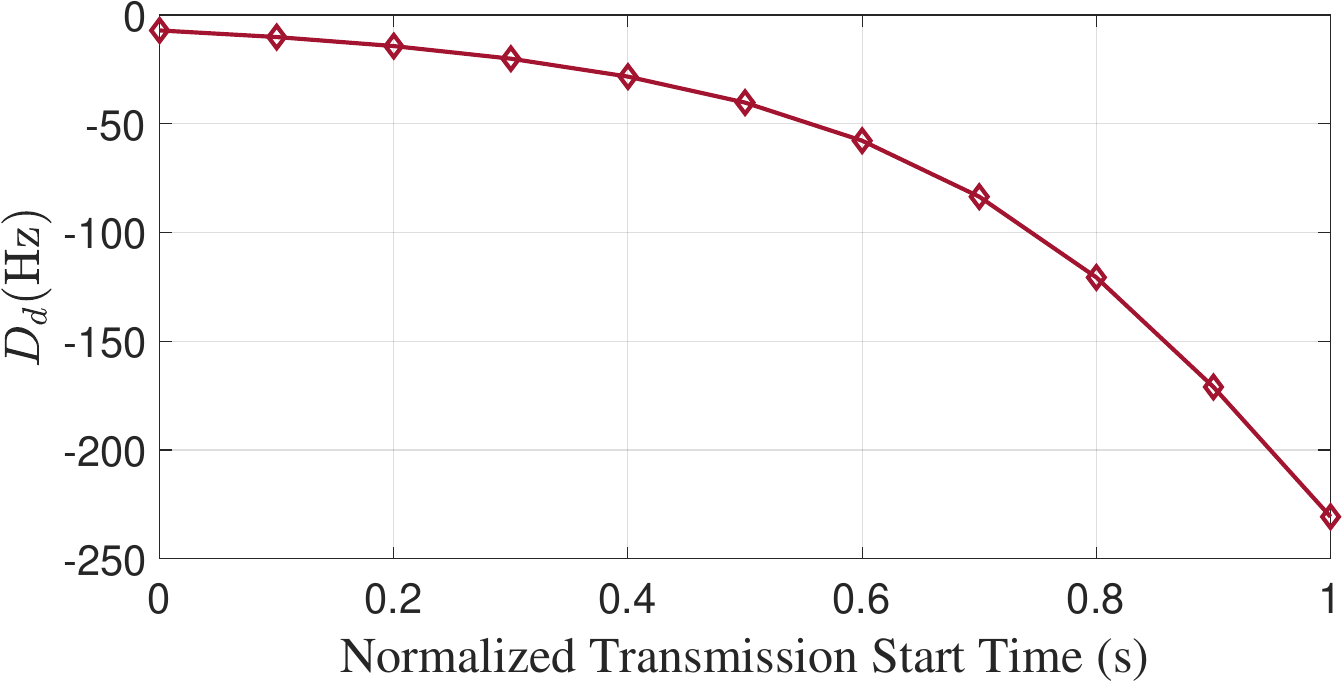}
  \caption{Differential Doppler shift during the first part of shared visibility window. The graph depicts the differential Doppler shift curve under the transition from a high Doppler shift to a high Doppler rate as a function of normalized transmission start time.}
  \label{Dd_over_window}
\end{figure}

Secondly, we will establish the relationship between $D_d$ and the carrier frequency difference. Based on the previously defined Doppler shift formula in \eqref{Doppler}, our simulations indicate that the Doppler shift can be reasonably approximated as a constant equal to the initial Doppler shift $v_d$ over one symbol period. Taking Fig. \ref{approximation} as an example, over one symbol period, the range of Doppler shift is approximately $5$ Hz, which is much smaller than the magnitude of Doppler shift $19$ kHz. Therefore, by treating it as a constant, Fig. \ref{carrier_shift} is plotted to show the impact of Doppler shift on the LoRa signal frequency. 
\begin{figure}[ht]
  \centering
  \includegraphics[width=8cm]{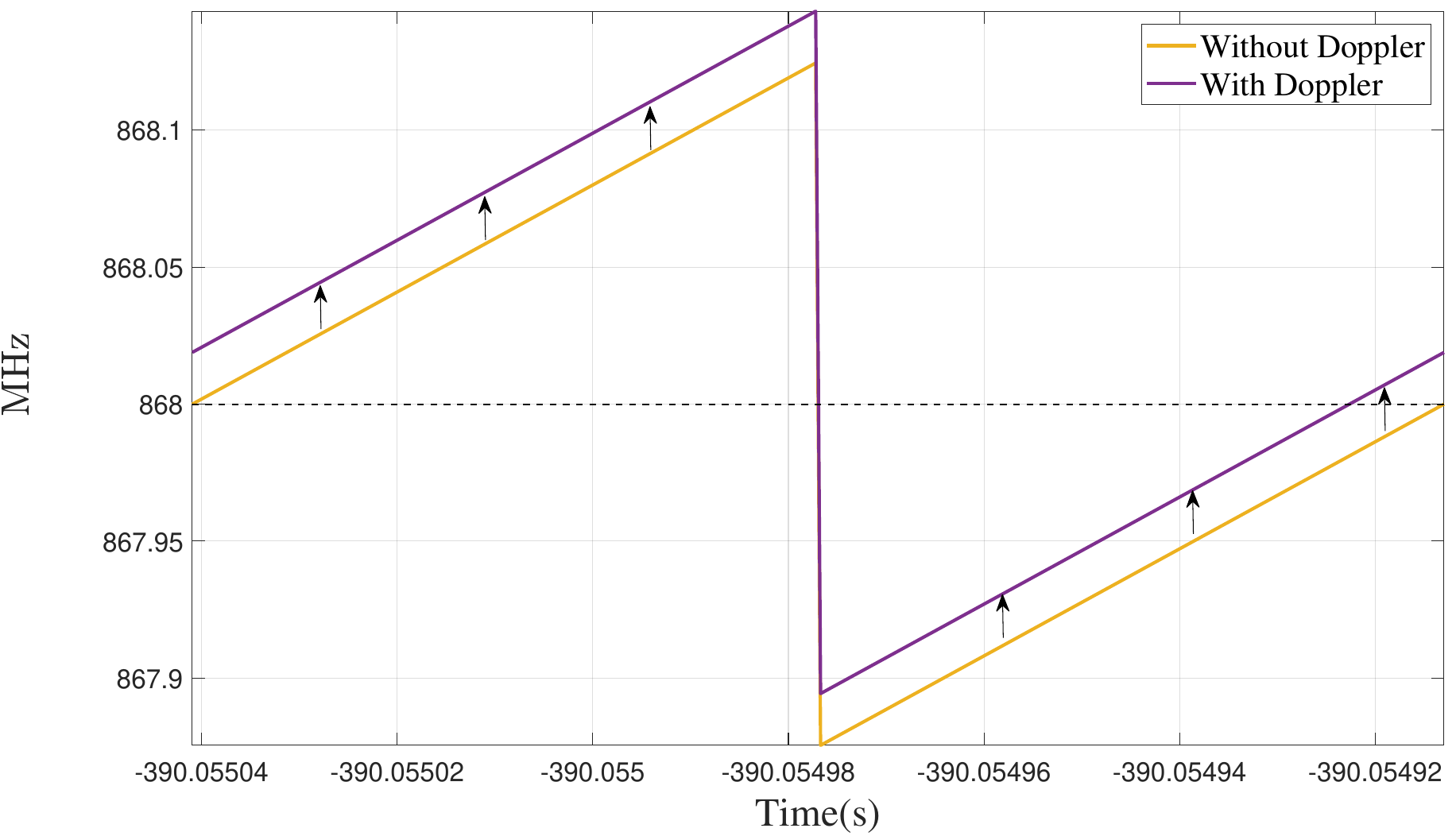}
  \caption{LoRa signal frequency with and without Doppler shift within one chirp period, with carrier frequency 868 MHz and bandwidth 250 kHz}.
  \label{carrier_shift}
\end{figure}

In Fig. \ref{carrier_shift}, the black dashed line shows the initial carrier frequency $868$ MHz. The purple line consistently surpasses the gold line, indicating a Doppler-induced frequency shift. By slightly raising the gold line, as indicated by arrows, it will align with the purple line. This figure implies that the Doppler-induced impact on the frequency curve can also be achieved by introducing a carrier frequency offset (CFO) equal to the constant Doppler shift value $v_d$. Furthermore, we can conclude that the Doppler-induced impact on the cross-correlation between two LoRa signals will be achieved by introducing a carrier frequency difference on these two signals, and this difference equals to $f_{c}+v_{d1} - (f_{c}+v_{d2}) = v_{d1} - v_{d2} = D_d$. 


Thirdly, we introduce the relationship between carrier frequency difference and maximum cross-correlation. From \eqref{cross-con} and \eqref{cross-dis}, we have simulated and obtained the maximum or mean cross-correlation results given varying carrier frequency differences. Based on these results, we have observed that the mean cross-correlation shows no variation given varying carrier frequency differences, so it is immune to the Doppler shift, which aligns perfectly with the observations we had in Figs. \ref{with_doppler_con}-(b), \ref{with_doppler_con}-(d), \ref{with_doppler_dis}-(b), and \ref{with_doppler_dis}-(d). However, the obvious fluctuations are observed for maximum cross-correlation. Therefore, in what follows, we will focus on the maximum cross-correlation results.

For the results when $SF_1 \neq SF_2$, without losing generality, we provided eight examples in Fig. \ref{carrier_impact_diff_SF}.
 \begin{figure}[ht]
  \centering
  \includegraphics[width=8.4cm]{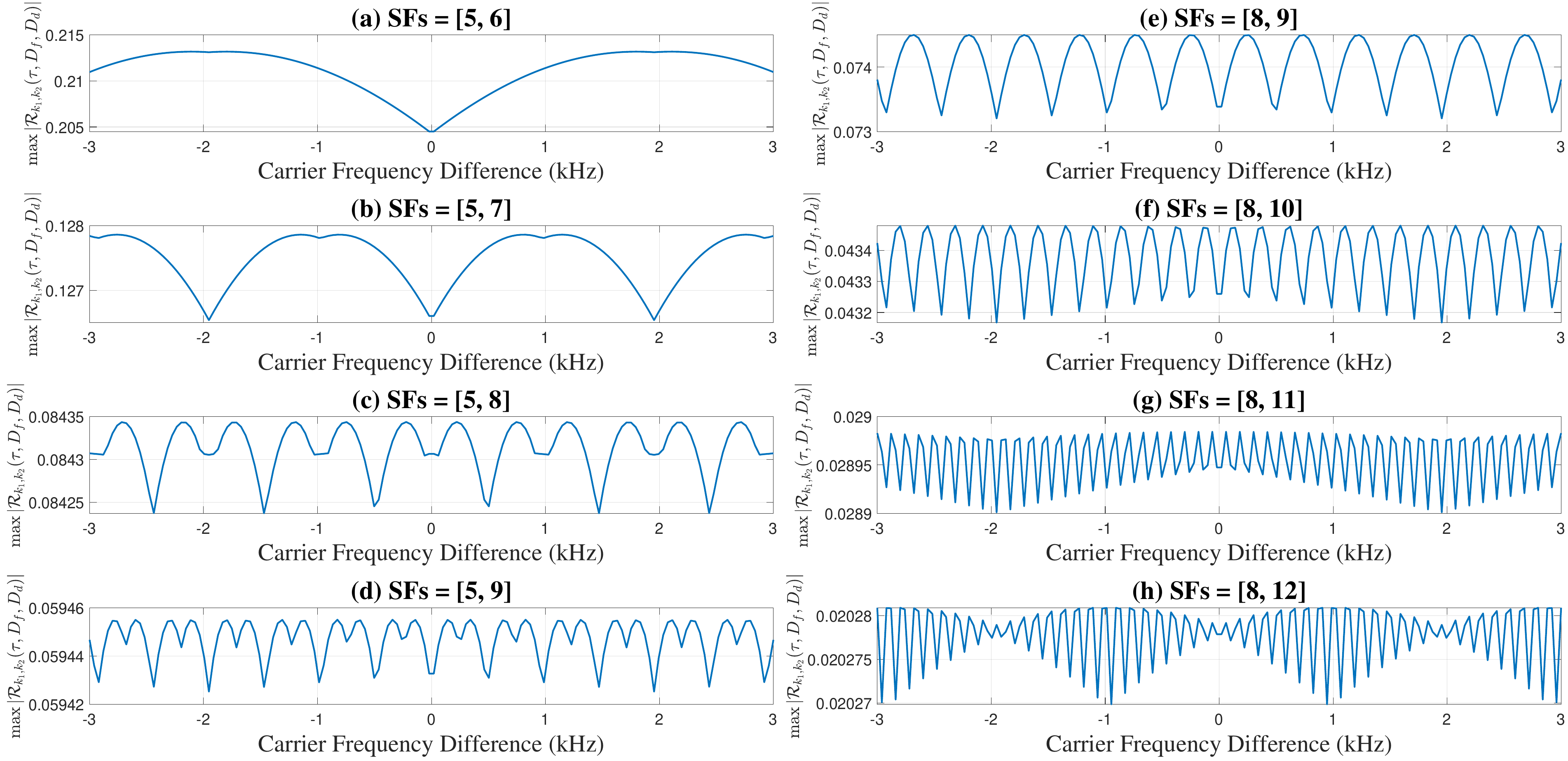}
  \caption{Carrier frequency difference’s impact on maximum cross-correlation when $SF_1 \neq SF_2$.}
  \label{carrier_impact_diff_SF}
\end{figure}
It is evident that although the maximum cross-correlation fluctuates with variations in carrier frequency difference, this fluctuation remains within a small range, which aligns well with the observed negligible impact of the Doppler effect in Figs. \ref{with_doppler_con}-(a), \ref{with_doppler_con}-(c), \ref{with_doppler_dis}-(a), and \ref{with_doppler_dis}-(c) when $SF_1 \neq SF_2$. 

For the results when $SF_1 = SF_2$, we plot the results with $SF$ from 5 to 12 in Fig. \ref{carrier_impact_all_same_SF}.
\begin{figure}[ht]
  \centering
  \includegraphics[width=8.4cm]{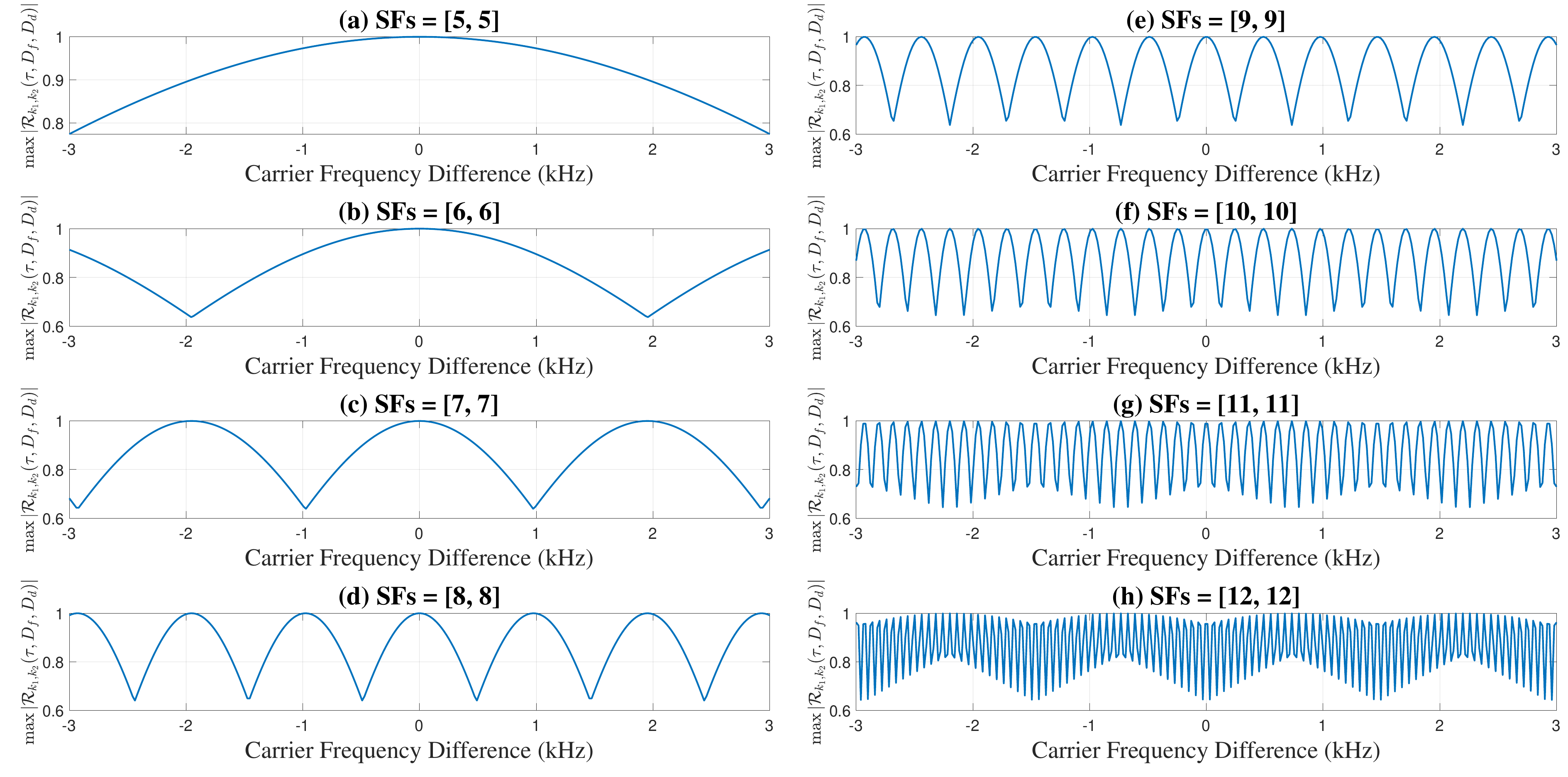}
  \caption{Carrier frequency difference’s impact on maximum cross-correlation when $SF_1 = SF_2$.}
  \label{carrier_impact_all_same_SF}
\end{figure}
In this figure, the maximum cross-correlation will mainly fluctuate between $0.6$ and $1$, which also aligns well with the decreased cross-correlation values observed in Fig. \ref{with_doppler_con}(c) and  Fig. \ref{with_doppler_dis}(c) when $SF_1 = SF_2$. Meanwhile, compared with the no-Doppler case, the maximum cross-correlations show a small variation when the carrier frequency difference is close to $0$ Hz. When the carrier frequency difference increases, large variations will occur periodically. Specifically, based on Fig. \ref{Dd_over_window}, the transmission start time with a high Doppler shift, with a carrier frequency difference close to $0$ Hz, will cause tiny changes in the maximum cross-correlation. However, the transmission start time with a high Doppler rate, with a carrier frequency difference around $-230$ Hz, will cause a large variation in the maximum cross-correlation. Therefore, it indicates that the maximum cross-correlation when $SF_1 = SF_2$ has good immunity to the high Doppler shift but no immunity to the high Doppler rate. Moreover, as the $SF$ increases, the curve becomes narrower. It can also be concluded that if there are two $SF$s: $m$ and $n$ $(m<n)$, the period of $SF$ = $n$ curve will be $\frac{1}{2^{n-m}}$ of the period of $SF = m$ curve. Thus, the results of high $SF$ are more sensitive to the Doppler effect than low $SF$.

Finally, based on the three relations mentioned above, we can establish the relationship from transmission start time to cross-correlation. To illustrate it, we examine an example for $SF_1 = SF_2 = 9$. 
If we choose the transmission start time in Fig. \ref{Dd_over_window} that obtained $D_d=-230$ Hz, the high Doppler rate will occur. 
According to Fig. \ref{carrier_impact_all_same_SF}-(e), with the carrier frequency difference equaling $-230$ Hz, the maximum cross-correlation will be $0.67$. This matches the results in Fig. \ref{with_doppler_con}-(c) and  Fig. \ref{with_doppler_dis}-(c) perfectly and confirms the rationality and feasibility of our explanations above.

\subsection{Parameter Analysis}\label{Parameter_analysis}
As the maximum cross-correlation when $SF_1 = SF_2$ is sensitive to Doppler shift, in this section, we will investigate its variation under various scenario parameters, including orbit height, orbit inclination, and the ground device distance. With our established relationship between transmission start time and maximum cross-correlation, we illustrate the curves of maximum cross-correlation when $SF_1 = SF_2$, and the curves of $D_d$ when transmission start time varies, given different scenario parameters and $SF$s. 

\subsubsection{Orbit Height} 
In Fig. \ref{H_final}, we plot the maximum cross-correlation and $D_d$ for different orbit heights under different $SF$s. 
In Fig. \ref{H_final}-(b), from high Doppler shift to high Doppler rate, the curves of $D_d$ with eight different orbit heights exhibit a consistently decreasing trend, and the disparity in $D_d$ among these orbit heights steadily increases. Notably, lower orbit height gets a higher absolute value $|D_d|$. Hence, it can be concluded that $D_d$ is quite sensitive to the variations in orbit height.

\begin{figure}[ht]
  \centering
  \includegraphics[width=8.4cm]{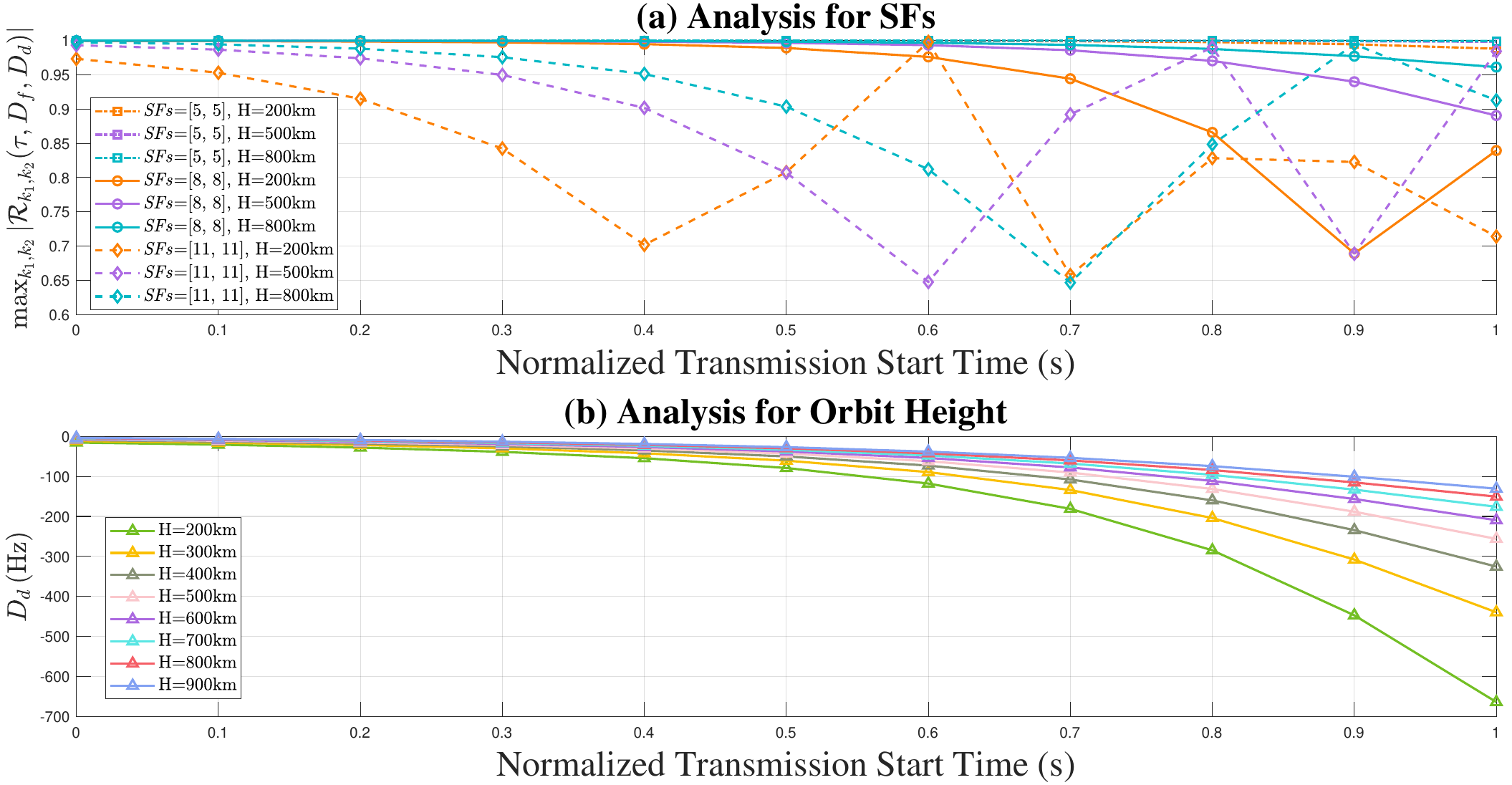}
  \caption{Parameter analysis for orbit height under different $SF$s.}
  \label{H_final}
\end{figure}
In Fig. \ref{H_final}-(a), from transmission start time with high Doppler shift to transmission start time with high Doppler rate, higher $SF$s hold more obvious fluctuations on the maximum cross-correlation than the lower $SF$s for a fixed orbit height. The maximum cross-correlation has more variation under low orbit height for fixed $SF$s. 

\subsubsection{Inclination Angle}
Similarly, we have plotted the maximum cross-correlation and $D_d$ under eight different inclination angles of orbit $i$ in Fig. \ref{i_final}-(a) and \ref{i_final}-(b). 
Unlike the orbit height, both maximum cross-correlation and $D_d$ are very close to each other even when the inclination angle varies.  
\begin{figure}[ht]
  \centering
  \centering
  \includegraphics[width=8.4cm]{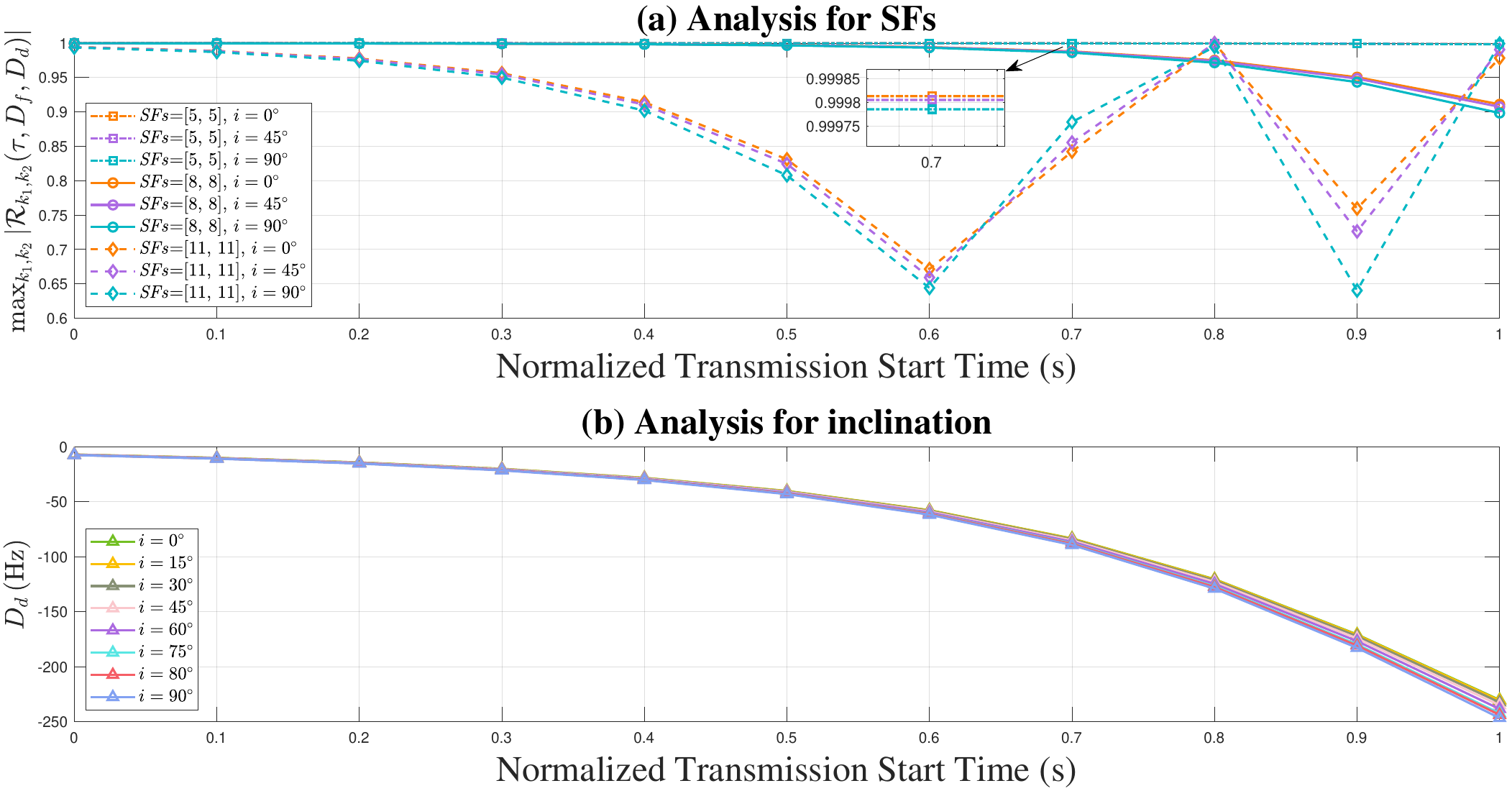}
  \caption{Parameter analysis for inclination angle under different $SF$s.}
  \label{i_final}
\end{figure}

Therefore, the maximum cross-correlation variation is nearly not affected by the inclination angle.

\subsubsection{Ground Device Distance}
In Fig. \ref{d_final}, we plotted the curves for different ground device distance $d$ under different $SF$s. 
Fig. \ref{d_final}-(b) shows that as $d$ gets longer, the $|D_d|$ is higher, which reaches 1200 \rm{Hz} at transmission start time with high Doppler rate when $d = 50$ \rm{km}. Fig \ref{d_final}-(a) shows that the maximum cross-correlation has more variation under long ground device distance.

\begin{figure}[ht]
  \centering
  \includegraphics[width=8.4cm]{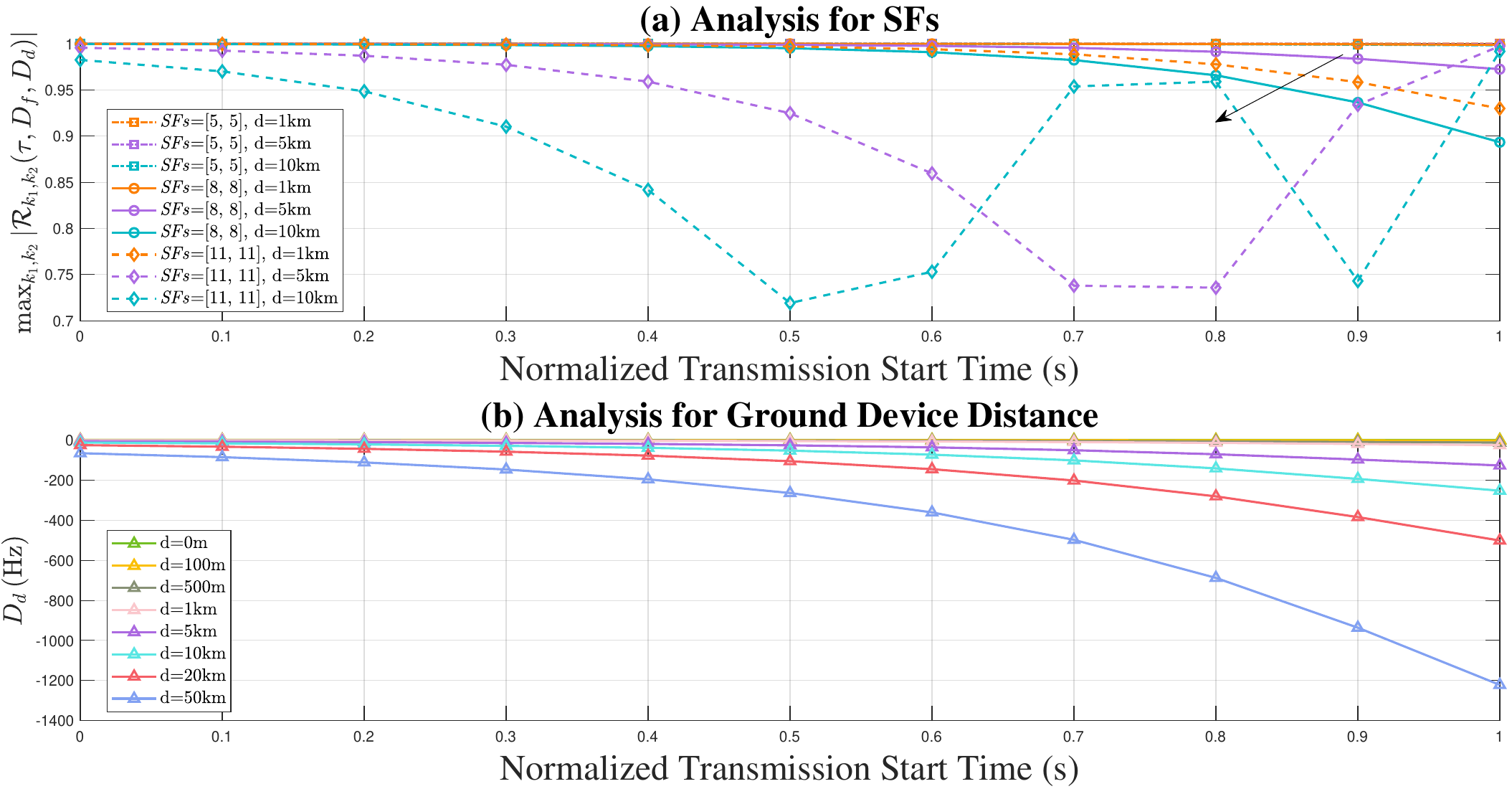}
  \caption{Parameter analysis for ground device distance under different $SF$s.}
  \label{d_final}
\end{figure}
It is worth noting that the fluctuation of cross-correlation is mainly based on the $D_d$ for a given normalized transmission start time, and the corresponding curve like Fig. \ref{carrier_impact_all_same_SF} for the given parameter setting. According to the above results, to minimize Doppler's impact on cross-correlation, it is recommended to use low $SF$s, high orbit height, short ground device distance, and transmission start time with high Doppler shift for transmission.

\subsection{BER Simulation}
In this part, we will analyze the impact of high Doppler shift or high Doppler rate on the BER. To make the analysis more comprehensive, we consider $\theta_{c\!A}=89^\circ$, $\theta_{c\!B}=89.2^\circ$, $[\theta_{\min \!A}=10^\circ, \theta_{\max \!A}=88^\circ]$ and $[\theta_{\min \!B}=10^\circ, \theta_{\max \!B}=88^\circ]$.

Firstly, we analyze the BER results for the LoRa signal $SF_1$ with Doppler shift and without interference signal under varying signal-to-noise ratio (SNR), as illustrated in Fig. \ref{BER_one_SF_Doppler}.  Under a high Doppler rate, as SNR increases, $SF_1 = 5,6,7$ get better BER performance, while $SF_1 >7$ gets no improvement, as shown in Fig. \ref{BER_one_SF_Doppler}-(a). Under high Doppler shift, all $SF$ have bad BER performance over the whole SNR range, as shown in Fig. \ref{BER_one_SF_Doppler}-(b). 
\begin{figure}[ht]
  \centering
  \includegraphics[width=8.4cm]{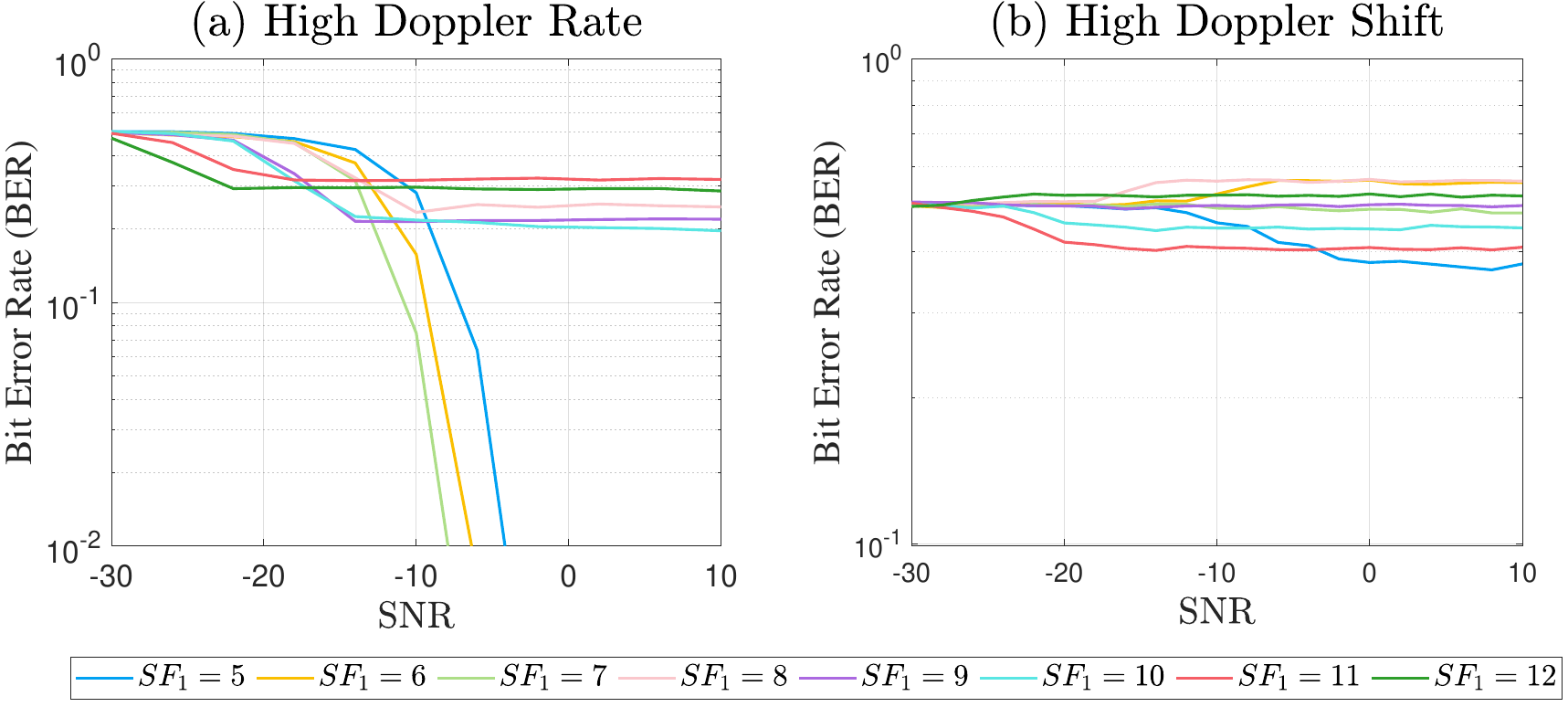}
  \caption{BER for $SF_1$ given different SNR with Doppler effect and without interference. (a) With high Doppler rate; (b) With high Doppler shift.}
  \label{BER_one_SF_Doppler}
\end{figure}

Secondly, we analyze the BER results for the LoRa signal $SF_1$ with Doppler shift and interference signal using $SF_2$ under varying signal-to-interference ratio (SIR). Taking $SF_1= 5$ and $SF_1= 10$ as an example, their BER performance given SNR = $0$ dB and different values of $SF_2$ are displayed in Fig. \ref{BER_two_SF_Doppler}. 
\begin{figure}[ht]
  \centering
  \includegraphics[width=8.4cm]{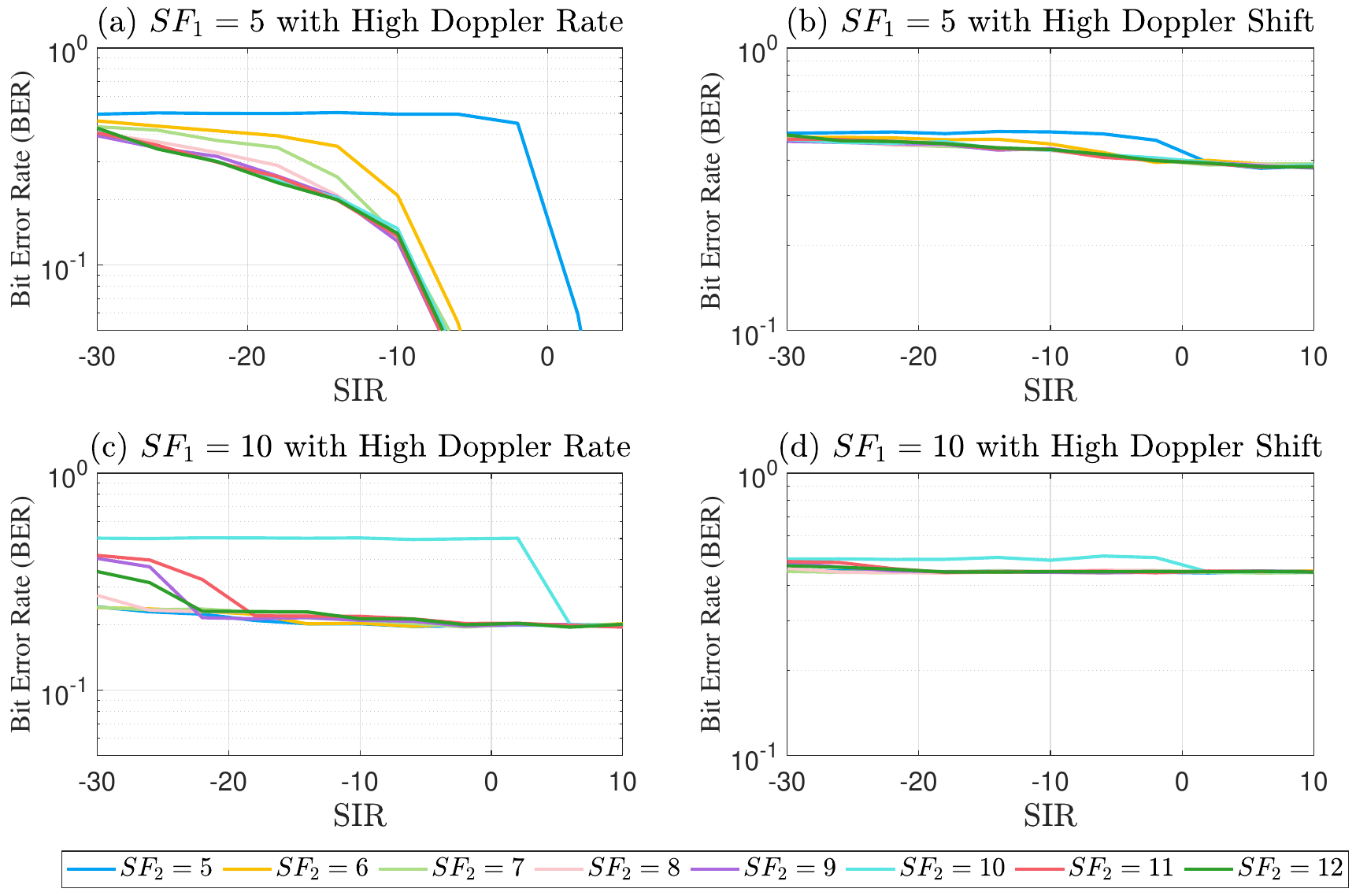}
  \caption{BER for $SF_1$ given different SIR with high Doppler shift or high Doppler rate and with interference $SF_2$.}
  \label{BER_two_SF_Doppler}
\end{figure}

In Fig. \ref{BER_two_SF_Doppler}, the worst BER performance of the LoRa signal is obtained by the case with interference from the signal using the same SF. The results of the no Doppler shift scenario in \cite{benkhelifa2022orthogonal} and \cite{al2021iot} have shown that BER gets better by shifting left when increasing the SF.
Comparing our results with the case of no Doppler, for $SF_1 = 5$ or $10$, under the high Doppler rate, BER performs slightly worse under interference. 
Under the high Doppler shift, we got the worst BER performance for all values of SIR. 
We can conclude that when increasing $SF_1$, BER will get worse. 
From a high Doppler rate to a high Doppler shift, the overall BER performance will deteriorate. 

To interpret these observations in Figs. \ref{BER_one_SF_Doppler} and \ref{BER_two_SF_Doppler}, we believe that their behaviors are related to the tolerable frequency shift threshold introduced in \cite{cao2021influence,ben2022new}, which can be calculated by $B/{2^{SF+1}}$.
Specifically, when the Doppler shift is below the threshold, increasing the SIR or SNR will help to improve the BER performance. Otherwise, the increase of the SIR or SNR will just contribute to little or no improvement in the BER performance, as observed in Figs. \ref{BER_two_SF_Doppler}-(b) and \ref{BER_two_SF_Doppler}-(d).

In \cite{benkhelifa2022orthogonal}, the authors concluded that the BER trend is aligned with the maximum cross-correlation values. However, we find that this alignment is no longer valid when the Doppler shift exists. Specifically, according to Fig. \ref{Corr_no_Doppler_Con_Dis}, \ref{with_doppler_con}, and \ref{with_doppler_dis}, from no-Doppler case to the case with the high Doppler shift, the maximum cross-correlations remained the same, but the BER performance became much worse. From no-Doppler case to the case with a high Doppler rate, the maximum cross-correlation decreased, but its BER performance was even better than that with a high Doppler shift. 

Therefore, to obtain promising BER performance, it is significant to ensure the Doppler shift is below the threshold. On one hand, low $SF$ and high bandwidth can be chosen to ensure a high threshold. On the other hand, a transmission start time with a high Doppler rate can be chosen to reduce the Doppler shift value. However, for the immunity of cross-correlation, we have recommended a different choice, which is transmission start time with a high Doppler shift. These conflicting suggestions indicate that it is necessary to design Doppler shift compensation techniques for conditions with high Doppler shift or high Doppler rate to maintain the orthogonality of LoRa modulation and to satisfy good BER performance. 

\section{Conclusion} \label{conclusion}
In this paper, we provided a detailed method for positioning the relative motion and locations between one moving LEO satellite and two stationary ground devices in a 3D system model. We proposed, for the first time, the shared visibility window for satellites towards two ground devices, and the general expression of LoRa waveform as well as the cross-correlations with Doppler shift in both continuous and discrete time domains. 

To further exhibit the underlying mechanism behind the cross-correlation results' trend affected by the Doppler effect, we proposed a theoretical explanation strategy, which highlights the relationship between transmission start time, differential Doppler shift, carrier frequency difference, and cross-correlation.
Comparing the with-Doppler and no-Doppler cases, we found the maximum cross-correlation for $SF_1 \!\neq\! SF_2$ and the mean cross-correlation are immune to the Doppler effect. But the maximum cross-correlation for $SF_1\!=\!SF_2$ is only immune to the high Doppler shift but not to the high Doppler rate. With the parameter analysis based on the maximum cross-correlation, we figured out that to improve the immunity of cross-correlation towards the Doppler effect, it is recommended to use low $SF$s, high orbit height, short ground device distance, and transmission start time with high Doppler shift for transmission.

Moreover, we noted that that BER under Doppler shift with or without interference is more sensitive to high Doppler shift. BER performance improves with increasing SIR or SNR only when the Doppler shift remains below the tolerable frequency shift threshold. To enhance the immunity of BER to the Doppler effect, the Doppler shift value is expected not to exceed the tolerable frequency shift threshold. Therefore, low $SF$, high bandwidth, and transmission start time with a high Doppler rate are recommended in this case. 

Last but not least, we found that the alignment of BER and maximum cross-correlation behaviors is no longer valid, and the recommendations on the transmission start time to enhance the immunity based on cross-correlation and BER analysis are controversial. These phenomena highlight the necessity of the Doppler compensation method to guarantee LoRa's good performance in satellite communication. In future work, we will explore the impact of time delay on cross-correlation and develop a feasible Doppler shift compensation scheme. Additionally, we plan to conduct real-world experiments to assess LoRa’s suitability for satellite communications. Meanwhile, as more and more satellites are deployed in space, we plan to extend this analysis work to the communication scenario with multiple satellites operating together to serve the ground devices.

\bibliographystyle{IEEEtran}
\bibliography{IEEEabrv,ref}

\end{document}